\documentclass[journal]{IEEEtran}
\usepackage{amsmath,amsfonts}
\usepackage{array}
\usepackage[caption=false,font=normalsize,labelfont=sf,textfont=sf]{subfig}
\usepackage{textcomp}
\usepackage{stfloats}
\usepackage{caption}
\usepackage{url}
\usepackage{verbatim}
\usepackage{amsthm,amsmath,amssymb}
\usepackage{mathrsfs}
\usepackage{autobreak}
\newtheorem{theorem}{Theorem}
\newtheorem{definition}{Definition}

\newtheorem{lemma}{Lemma}

\usepackage{enumitem}
\usepackage{color}
\definecolor{lblue}{RGB}{75, 145, 209}
\definecolor{lyellow}{RGB}{191,144,0}
\definecolor{lred}{RGB}{241,89,96}
\definecolor{lgreen}{RGB}{112,173,71}
\usepackage{graphicx}

\usepackage{float}
\usepackage{hyperref}
\usepackage{algorithmic}
\usepackage{algorithm}
\usepackage{cite}
\usepackage[normalem]{ulem} 

\newdimen\Lmargin
\newdimen\Rmargin
\Rmargin=\paperwidth \advance\Rmargin by-\Lmargin \advance\Rmargin by-\textwidth
\ifdim\Lmargin>\Rmargin \Rmargin=\Lmargin \fi

\usepackage{booktabs}
\usepackage{multirow}
\usepackage{pifont}  
\usepackage{ulem} 
\usepackage{etoolbox} 
\makeatletter
\patchcmd{\@makecaption}
  {\scshape}
  {}
  {}
  {}
\makeatother

\usepackage{bm}

\begin{document}
\title{Compressive Imaging Reconstruction via Tensor Decomposed Multi-Resolution Grid Encoding}
\author{Zhenyu Jin, Yisi Luo, Xile Zhao, Deyu Meng
\thanks{Zhenyu Jin, Yisi Luo, and Deyu Meng are with the School of Mathematics and Statistics, Xi'an Jiaotong University.} 
 \thanks{Xile Zhao is with the School of Mathematical Science, University of Electronic Science and Technology of China.}}
\maketitle
\begin{abstract}
Compressive imaging (CI) reconstruction, such as snapshot compressive imaging (SCI) and compressive sensing magnetic resonance imaging (MRI), aims to recover high-dimensional images from low-dimensional compressed measurements. This process critically relies on learning an accurate representation of the underlying high-dimensional image. However, existing unsupervised representations may struggle to achieve a desired balance between representation ability and efficiency. To overcome this limitation, we propose Tensor Decomposed multi-resolution Grid encoding (GridTD), an unsupervised continuous representation framework for CI reconstruction. GridTD optimizes a lightweight neural network and the input tensor decomposition model whose parameters are learned via multi-resolution hash grid encoding. It inherently enjoys the hierarchical modeling ability
of multi-resolution grid encoding and the compactness of tensor decomposition, enabling effective and efficient reconstruction of high-dimensional images. Theoretical analyses for the algorithm's Lipschitz property, generalization error bound, and fixed-point convergence reveal the intrinsic superiority of GridTD as compared with existing continuous representation models. Extensive experiments across diverse CI tasks, including video SCI, spectral SCI, and compressive dynamic MRI reconstruction, consistently demonstrate the superiority of GridTD over existing methods, positioning GridTD as a versatile and state-of-the-art CI reconstruction method.
\end{abstract}
\begin{IEEEkeywords}
Snapshot compressive imaging, MRI reconstruction, multi-resolution grid encoding, tensor decomposition. 
\end{IEEEkeywords}
\section{introduction}
Compressive imaging (CI) reconstruction, aiming to recover high-dimensional images from the corresponding low-dimensional compressed measurements, serves as an important paradigm for numerous imaging applications, such as medical diagnostics \cite{gamper2008compressed}, remote sensing \cite{cao2016computational}, and computational photography \cite{llull2013coded}. Despite its significance, CI reconstruction remains inherently ill-posed due to the irreversible information loss during compression. Traditional methods, such as optimization-based approaches with handcrafted priors (e.g., sparsity \cite{zha2023learning} or total variation \cite{yuan2016generalized}), may struggle to capture the complex underlying structure of the data, especially in high-dimensional settings. Meanwhile, supervised deep learning methods \cite{yuan2021snapshot,qiao2020deep,heckel2024deep} rely on paired training data, which may limit the applicability in scenarios where ground-truth samples are scarce or unavailable. To overcome these limitations, unsupervised learning approaches have garnered increasing attention in the field of CI reconstruction \cite{qayyum2022untrained,alkhouri2025understanding,akccakaya2022unsupervised}. These methods learn a deep representation of the underlying high-dimensional image by solely using the observed measurement in an unsupervised manner, which eliminates the need for paired training datasets by leveraging the intrinsic structure prior of the image. Representative unsupervised CI reconstruction methods focus on training a randomly initialized deep neural network with physical constraints to represent the high-dimensional image. The most classical methods would be the deep image prior (DIP)-based approaches \cite{DIP_IJCV}, which map random codes to the high-quality data through deep convolutional neural networks (CNNs). For instance, Yoo et al. \cite{yoo2021time} proposed a time-dependent DIP for dynamic magnetic resonance imaging (MRI) reconstruction. Zhao et al. \cite{zhao2024untrained} proposed a bagged deep video prior method to enhance the performance of untrained neural networks for snapshot compressive imaging. More DIP-based approaches can be found in relevant studies \cite{miao_tgrs,S2DIP,lrsdn}. 
\par
Recently, the implicit neural representation (INR)\cite{sitzmann2020implicit} has emerged as an alternative unsupervised framework for CI reconstruction, which leverages coordinate-based neural networks to parameterize the image as a continuous function, serving as a more lightweight representation structure. The INRs map spatial coordinates of the image to their corresponding values, learning a continuous and compact representation for CI reconstruction. For example, Sun et al. \cite{coil} proposed the CoIL, an unsupervised coordinate-based framework that learns a mapping from measurement coordinates to responses using INR, enabling high-fidelity measurement generation to enhance various image reconstruction methods for sparse-view computed tomography (CT). Shen et al. \cite{shen2022nerp} presented the NeRP, an INR framework for sparse image reconstruction that leverages a single prior image and measurement physics while generalizing to CT/MRI and detecting subtle tumor changes. Nevertheless, conventional INRs often suffer from insufficient representation abilities due to the spectral bias of neural networks towards low-frequency components\cite{jacot2018neural, tancik2020fourier}, which restricts their ability to capture fine-grained high-frequency details in the image.\par 
To overcome these challenges, recent advancements based on grid encoding models have been studied for continuous representation. As a representative example, the instant neural graphics primitive (InstantNGP) \cite{muller2022instant} introduces a novel structure that integrates multi-resolution grids stored in hash tables with a lightweight neural network to learn a continuous representation of data. The InstantNGP enhances the representation ability for high-frequency details as compared with INRs, making the continuous representation more powerful and accurate for unsupervised compressive imaging reconstruction. For instance, Feng et al. \cite{feng2025spatiotemporal} proposed an InstantNGP-based spatial-temporal method for unsupervised compressive dynamic MRI reconstruction, which holds superior representation abilities.\par 
Nevertheless, grid encoding-based frameworks (e.g., InstantNGP) still face limitations in scalability and efficiency, particularly for large-scale, high-dimensional images. A key issue is their significant parameter and computational overhead. The InstantNGP relies on multi-resolution {\bf high-dimensional grids} \cite{muller2022instant,feng2025spatiotemporal}, causing model parameters and computational complexity to grow dramatically when increasing data dimensionality (see Table \ref{tab:complexity} for example). 
Hence, how to effectively balance the representation capacity and efficiency of grid encoding models remains a critical challenge for high-dimensional CI reconstruction. \par
To overcome this challenge, we propose a novel Tensor Decomposed multi-resolution Grid encoding model (termed GridTD), an unsupervised continuous representation framework for CI reconstruction. GridTD leverages a compact tensor decomposition parameterized by multi-resolution grid encoding, which utilizes {\bf lightweight one-dimensional grids} instead of high-dimensional grids, and the one-dimensional grids can be stored in smaller hash tables. Then, a lightweight neural network is leveraged to reconstruct the high-dimensional image. GridTD enjoys both the hierarchical modeling ability of multi-resolution grid encoding and the compactness of tensor decomposition to achieve a compact and scalable representation, enabling more efficient unsupervised CI reconstruction. Especially, the parameter number of GridTD scales linearly w.r.t. the data dimension, while traditional InstantNGP scales exponentially. Moreover, GridTD enhances the robustness and reconstruction performance by leveraging the structure constraint brought by tensor decomposition. 
 The contributions of this work are summarized as follows:
\begin{itemize}
    \item We propose GridTD, a novel continuous representation framework using tensor decomposed multi-resolution grid encoding, designed for effective and efficient compressive imaging reconstruction for high-dimensional images, including video snapshot compressive imaging (SCI), spectral SCI, and compressive dynamic MRI reconstruction.
    \item We establish rigorous theoretical analyses, including the Lipschitz continuity bound and generalization error bound of the GridTD model, and the fixed-point convergence for the alternating optimization scheme induced by GridTD, to theoretically validate the superiority of our method as compared with existing models.
    \item We further introduce tailored regularization methodologies for GridTD, including the temporal affine adapter and smooth regularization terms to enhance the dynamic modeling capability of GridTD for irregular high-dimensional data such as temporal motions of videos.
    \item Extensive experiments across diverse CI tasks (video SCI, spectral SCI, and dynamic MRI) demonstrate that GridTD outperforms existing unsupervised and continuous representation methods with lower computational costs, positioning GridTD as a versatile and state-of-the-art method for CI reconstruction.
\end{itemize}\par  
The remainder of this paper is organized as follows: Section II reviews some preliminaries and related works. Section III details the GridTD framework, including model structures, theoretical analyses, and optimization strategies. Section IV presents experimental results. Section VI concludes the paper.
\section{Preliminaries and Related Works}
\subsection{Notations}\label{subsec:nota}
In this paper, we use $x, {\bf x}, {\bf X}, \mathcal{X}$ to denote scalar, vector, matrix, and tensor, respectively. The $i$-th element of ${\bf x}$ is denoted by ${\bf x}[i]$, and it is similar for matrices and tensors, e.g., ${\cal X}[i_1,\cdots,i_n]$. The $\lfloor\cdot\rfloor$ denotes the rounded down operation. 
Vector operations include the concatenation $\oplus$, the outer product $\otimes$, and the Hadamard (element-wise) product $\odot$. For example, given two vectors, the concatenation and outer product are defined as
$$\small
\begin{pmatrix} {\bf x}_1\\ {\bf x}_2 \end{pmatrix} \oplus \begin{pmatrix}{\bf y}_1 \\ {\bf y}_2 \end{pmatrix}  =  \begin{pmatrix} {\bf x}_1\\ {\bf x}_2 \\{\bf y}_1\\{\bf y}_2\end{pmatrix},\;\begin{pmatrix} {\bf x}_1\\ {\bf x}_2 \end{pmatrix} \otimes \begin{pmatrix} {\bf y}_1 \\ {\bf y}_2 \end{pmatrix}  =  \begin{pmatrix} {\bf x}_1 {\bf y}_1,\;{\bf x}_1 {\bf y}_2 \\ {\bf x}_2 {\bf y}_1,\;{\bf x}_2 {\bf y}_2\end{pmatrix}.
$$ 
Note that the outer product $\otimes$ can be analogously extended to multiple vectors and return a high-order tensor. The canonical polyadic (CP) tensor decomposition \cite{SIAM_review} for a $D$-th order tensor ${\cal X}$ is formulated as
\begin{equation}\small\label{CP}
{\cal X}=\left(\sum_{r=1}^R \bigotimes_{d=1}^D {\bf h}_{r,d}\right)\in{\mathbb R}^{n_1\times \cdots \times n_D},
\end{equation}
where ${\bf h}_{r,d}\in{\mathbb R}^{n_d}$ is the factor vector. The $\bigotimes_{d=1}^{D}$ denotes the iterated outer product for $D$ vectors, which returns a $D$-th order tensor. Throughout the manuscript, the notations $D$ and $d=1,2,\cdots,D$ are used for dimension indices, and the notations $n_d$ and $i=1,2,\cdots,n_d$ are used for tensor indices.
\subsection{Snapshot Compressive Imaging}
The snapshot compressive imaging (termed SCI) is an important imaging technique that records high-dimensional information onto a 2D sensor in a single exposure by leveraging creative optical encoding schemes \cite{NC}.  The resulting compressed measurement is then computationally decoded to recover the full spatial-temporal or spectral data cube. Based on the type of data captured, SCI can be broadly classified into two groups: video SCI and spectral SCI. The video SCI captures high-speed video scenes by compressing multiple frames into a single snapshot measurement. It then leverages compressive sensing and optimization algorithms to reconstruct the original video sequence from the encoded data \cite{yuan2021plug}. The degradation model of video SCI is formulated as
\begin{equation}\small
\mathbf{Y} = \sum_{t=1}^{n_3} \mathcal{M}_t \odot \mathcal{X}_t + \mathbf{Z},\label{eq:sci_model_alt}
\end{equation}
where the original video sequence $\mathcal{X} \in \mathbb{R}^{n_1 \times n_2 \times n_3}$ is encoded through a set of modulation masks $\mathcal{M} \in \mathbb{R}^{n_1 \times n_2 \times n_3}$, with $\mathcal{M}_t$ denoting the $t$-th mask corresponding to the $t$-th frame $\mathcal{X}_t$. The measurement $\mathbf{Y} \in \mathbb{R}^{n_1 \times n_2}$ is obtained by summing the element-wise modulated frames, corrupted by additive noise $\mathbf{Z} \in \mathbb{R}^{n_1 \times n_2}$. Similarly, the spectral SCI encodes multiple spectral bands of a hyperspectral image (HSI) using a coded aperture, enabling high-dimensional data reconstruction from a 2D snapshot while preserving spectral fidelity \cite{lrsdn}. The degradation model of the spectral SCI is formulated as
\begin{equation}\small
\label{eq:spectral_sci_model}
\begin{split}
\mathbf{Y} = \sum_{t=1}^{n_3}{\rm shift}(\mathcal{M}_t \odot \mathcal{X}_t) + \mathbf{Z},
\end{split}
\end{equation}
where $\mathcal{X} \in \mathbb{R}^{n_1 \times n_2  \times n_3}$ denotes the original HSI with $n_3$ spectral bands, ${\rm shift}(\mathcal{M}_t \odot \mathcal{X}_t)\in\mathbb{R}^{n_1 \times (n_2 + d(n_3 - 1))}$ denotes the $t$-th spatially shifted spectral band defined by ${\rm shift}(\mathcal{M}_t \odot \mathcal{X}_t)[i_1,i_2] = (\mathcal{M}_t \odot \mathcal{X}_t)[i_1,i_2 + d(t - 1)] $, and $\mathcal{M} \in \mathbb{R}^{n_1 \times n_2 \times n_3}$ represents the corresponding modulation mask. The measurement $\mathbf{Y} \in \mathbb{R}^{n_1 \times (n_2 + d(n_3 - 1))}$ is obtained by summing the shifted bands. Additive noise $\mathbf{Z}\in\mathbb{R}^{n_1\times(n_2 +d(n_3 - 1))}$ is introduced to model acquisition perturbations. The SCI reconstruction aims to reconstruct the high-dimensional image ${\cal X}$ from the measurement ${\bf Y}$. The SCI reconstruction algorithms can be categorized into three main groups: traditional optimization-based methods, supervised learning methods, and unsupervised learning methods. Early traditional approaches rely on handcrafted priors such as TV \cite{yuan2016generalized}, sparsity, and nonlocal similarity \cite{liu2018rank} to model data structure, which hold good interpretability and robustness.\par
With advancements in deep learning, data-driven approaches have gained attention. Supervised learning methods typically employ encoder-decoder architectures to learn priors directly from training data. For instance, Wang et al. \cite{wang2025s} proposed a spatial-spectral transformer with parallel attention architecture and mask-aware learning for video SCI. Cao et al. \cite{cao2024hybrid} proposed the EfficientSCI++, an efficient hybrid CNN-transformer network to achieve efficient reconstruction with lower computational costs for video SCI. Zhang et al. \cite{zhang2024dual} proposed the dual prior unfolding framework, which integrates dual priors and a PCA-inspired focused attention mechanism for spectral SCI. However, supervised methods may face limitations when the original high-dimensional data cannot be acquired or when testing samples deviate from the training distribution. Unsupervised methods address these challenges by recovering the high-dimensional data solely from compressed measurements without paired training examples. The plug-and-play (PnP) framework integrates learned deep priors into optimization pipelines without retraining. For example, Yuan et al. \cite{yuan2021plug} proposed the PnP-ADMM and PnP-GAP, which leverage deep video denoising priors to enable high-quality reconstruction for video SCI. The DIP \cite{DIP_IJCV} takes a different approach by using the network architecture itself as an implicit prior. For instance, Miao et al.\cite{miaotci} proposed the FactorDVP model, which decomposes video data into foreground and background components, and models these components through separated deep video priors for video SCI. Zhao et al. \cite{zhao2024untrained} proposed the bagged deep video prior using untrained neural networks for video SCI. 
Diffusion posterior sampling \cite{chung2023diffusion} is another type of unsupervised method, which modifies the posterior sampling process of diffusion models without requiring network retraining. For instance, Pang et al. \cite{pang2024hir} proposed the HIR-Diff, an unsupervised hyperspectral restoration framework that integrates diffusion posterior sampling with TV prior and low-rank decomposition for HSI reconstruction. Miao et al. \cite{miao2023dds2m} proposed the DDS2M, a self-supervised diffusion model for HSI restoration, which leverages a variational spatio-spectral module to infer posterior distributions during reverse diffusion, enabling training-free adaptation to degraded HSIs.
Tensor decomposition is also considered an unsupervised learning method. Luo et al. \cite{luo2022hlrtf} proposed the HLRTF, a hierarchical low-rank tensor factorization method to capture multi-dimensional image structures for spectral SCI. To our knowledge, we are the first to leverage the continuous representation framework for the SCI problem.
Our continuous GridTD representation leverages a novel tensor decomposition model parameterized by multi-resolution grid encoding, which enjoys stronger representation abilities brought by the multi-resolution grid encoding and favorable computational efficiency from the compact tensor decomposition. 
\subsection{Compressive Sensing Dynamic MRI}
Dynamic MRI provides high tissue contrast while capturing real-time physiological changes, making it essential for imaging moving organs like the heart and abdomen. However, its clinical use faces a trade-off between spatial and temporal resolution due to scan time limitations. To overcome this, compressive imaging techniques using undersampled $k$-space acquisition have been developed \cite{qin2018convolutional,huang2021deep}. The compressive sensing dynamic MRI reconstruction aims to recover an high-quality, time-resolved MRI $\mathbf{X}$ from the undersampled $k$-space data $\mathbf{Y}_c$ by exploiting spatial-temporal correlations and prior knowledge \cite{yoo2021time}. The acquisition process of the MRI can be modeled as
\begin{equation}\small\begin{small}
\mathbf{Y}_c = \mathbf{F}_u \mathbf{S}_c \mathbf{X}, \quad 1 \leq c \leq C,
\label{eq:mri_model}
\end{small}
\end{equation}
where $\mathbf{X} \in \mathbb{C}^{N^2 \times T}$ denotes a time-resolved MRI sequence consisting of $T$ frames, each with spatial dimensions $N \times N$. The signal $\mathbf{Y}_c \in \mathbb{C}^{NM \times T}$ corresponds to the undersampled $k$-space observation collected by the $c$-th receiver coil, where $M < N$ indicates the number of sampled phase-encoding lines per frame and $C$ is the total number of coils. The matrix $\mathbf{S}_c \in \mathbb{C}^{N^2 \times N^2}$ captures the spatial sensitivity profile of the $c$-th coil and is assumed to be diagonal. The sampling operator $\mathbf{F}_u \in \mathbb{C}^{NM \times N^2}$ represents the non-uniform Fourier transform, implemented via the non-uniform fast Fourier transform (NUFFT) algorithm. Recovering the dynamic MRI $\bf X$ from the undersampled data ${\bf Y}_c$ is challenging. The pioneer work proposed in \cite{feng2014golden} introduces a golden-angle radial sparse parallel MRI reconstruction framework by using temporal frame grouping and time-averaged coil sensitivity estimation with total variation (TV) regularization, inspiring many subsequent works for dynamic MRI reconstruction. \par 
Recent advances in deep learning have transformed the reconstruction pipeline of MRI. For example, Qin et al. \cite{qin2018convolutional} proposed a bidirectional convolutional recurrent neural network combining iterative algorithm principles for dynamic MRI reconstruction. Huang et al. \cite{huang2021deep} presented the deep low-rank plus sparse network, an unfolding method that injects priors such as learnable low-rank and sparse constraints directly into the network layers for MRI reconstruction. Unsupervised MRI reconstruction methods have also emerged in recent years. For instance, Yoo et al. \cite{yoo2021time} proposed an unsupervised deep image prior framework for MRI reconstruction, where an initial input is generated from a predefined manifold, mapped to a latent space through a neural network, and then passed through the DIP for optimization. Implicit neural representation (termed INR) \cite{tancik2020fourier,sitzmann2020implicit} has emerged as another promising unsupervised framework for dynamic MRI reconstruction, owing to its implicit continuous regularization. These methods typically learn continuous mappings from spatial or $k$-space coordinates to signal values. As representative examples, Huang et al. \cite{huang2023neural} proposed a Fourier multilayer perception (MLP)-based INR that maps $k$-space coordinates to complex $k$-space signals for MRI reconstruction. Kunz et al. \cite{kunz2024implicit} introduced a Fourier-MLP mapping from image coordinates to image intensities for MRI reconstruction. To further enhance efficiency and representational capacity, Feng et al. \cite{feng2025spatiotemporal} incorporated the InstantNGP model \cite{muller2022instant} into dynamic MRI reconstruction, significantly accelerating training while improving expressiveness of the network for fine details recovery. As compared with these unsupervised methods, this work proposes a novel tensor decomposed multi-resolution grid encoding model for continuous representation, which enjoys competitive representation abilities and better computational efficiency due to the compactness of tensor decomposition, thus achieving better performances and efficiency for MRI reconstruction.

\section{The Proposed Method}\label{sec:method}
\subsection{Tensor Decomposed Multi-Resolution Grid Encoding}\label{subsec:ten}
\begin{figure*}
    \centering
    \includegraphics[width=1\linewidth]{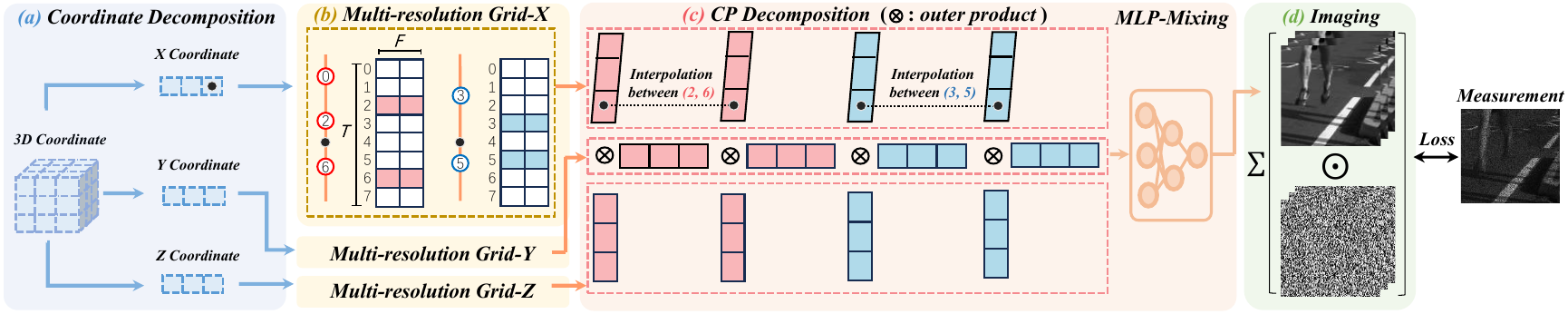}
    \caption{Illustration of the GridTD model for compressive imaging reconstruction. {\color{lblue}(a)} The inputs of the model are decomposed coordinate vectors of a tensor. {\color{lyellow} (b)} The lightweight 1D multi-resolution grid encoding (Eq. \eqref{eq_gridtd}) is employed to generate factor matrices of the CP tensor decomposition through linear interpolation for each coordinate (black dot) using surrounding multi-resolution 1D grids (colored circles in {\color{lyellow} (b)}). {\color{lred}(c)} Given the factor matrices obtained from multi-resolution grids, the CP decomposition is employed to generate the high-dimensional image with a lightweight MLP as the feature decoder. {\color{lgreen}(d)} Finally, the model output is optimized using an unsupervised loss related to the compressive sensing process.}
    \label{fig:gradient}\vspace{-0.3cm}
\end{figure*} 
The key point for CI reconstruction is to learn an accurate representation of the underlying high-dimensional image. Among existing methods, modeling the image as a continuous representation that maps a spatial coordinate ${\bf v}\in{\mathbb R}^D$ to the corresponding image value is an emerging and effective paradigm for CI reconstruction \cite{huang2023neural,kunz2024implicit,feng2025spatiotemporal}. This continuous modeling framework enjoys essential advantages such as implicit continuous regularization, lightweight neural structures, and resolution independence. Among the continuous representation structures, the multi-resolution grid encoding (or say InstantNGP) \cite{muller2022instant} stands out as a state-of-the-art approach, which constructs a continuous function using multi-resolution grid interpolation. InstantNGP efficiently stores the multi-resolution grids in multi-resolution hash tables. The function value of each input coordinate is queried by interpolating over the multi-resolution grids. The interpolated results from different resolutions of grids are concatenated into a feature vector, which is passed through a lightweight neural network to produce the desired output. InstantNGP 
 can capture high-frequency components and fine details of visual data better than conventional INRs by using learnable multi-resolution encoding, hence showcasing strong performance for continuous CI reconstruction \cite{feng2025spatiotemporal}. We formalize the InstantNGP as the multi-resolution grid encoding as follows.
\begin{definition}[Multi-resolution grid encoding \cite{muller2022instant}]\label{def:mrh}
Given a normalized input coordinate vector $\mathbf{v} \in [0,1)^D$, the multi-resolution grid encoding function ${\bf H}({\bf v}):[0,1)^{D}\rightarrow{\mathbb R}^{LF}$ with $L$ resolutions is defined as:
\begin{equation}\small
    \mathbf{H}(\mathbf{v}) := \left(\bigoplus_{l=1}^{L} \mathbf{h}_l(\mathbf{v})\right)\in{\mathbb R}^{LF},
\end{equation}
where the $l$-th resolution encoding $\mathbf{h}_l(\mathbf{v})\in{\mathbb R}^F$ is given by:
\begin{equation*}\small
\begin{split}
    &\mathbf{h}_l(\mathbf{v}) = \sum_{\mathbf{b} \in \{1,2\}^D} \mathcal{W}_{\mathbf{v},l}[\mathbf{b}] \cdot \mathcal{G}_{l}\left[ \lfloor  (N_l-1) \mathbf{v}\rfloor + \mathbf{b},: \right],\\
    &\mathcal{W}_{\mathbf{v},l} = \bigotimes_{d=1}^{D}\begin{pmatrix}1 - \mathbf{u}_{\mathbf{v},l} \left[d \right] \\ \mathbf{u}_{\mathbf{v},l }\left[d \right]\end{pmatrix},\;\mathbf{u}_{\mathbf{v},l} =  (N_l-1)\mathbf{v} - \lfloor  (N_l-1)\mathbf{v} \rfloor. 
    \end{split}
\end{equation*}
The components are defined as follows:
\begin{itemize}
    \item $\mathcal{G}_{l} \in \mathbb{R}^{N_l\times \cdots\times N_l\times F}$ is a $(D+1)$-order tensor representing the $l$-th grid encoding tensor, where $F$ denotes the dimension of each feature vector on grids and $N_l$ denotes the $l$-th grid resolution.
    \item $\mathcal{W}_{\mathbf{v},l}\in{\mathbb R}^{2\times\cdots\times 2}$ is a $D$-th order tensor representing the linear interpolation weights of the $l$-th grid conditioned on the input position $\bf v$.
    \item $\mathbf{u}_{\mathbf{v},l} \in [0,1)^D$ are relative coordinates of $\bf v$ w.r.t. its surrounding grids in ${\cal G}_l$ and ${\bf b}$ is the grid vertices indicator.
\end{itemize}
\end{definition}
{The philosophy of the InstantNGP is to use multi-resolution grids $\{\mathcal{G}_{l}\}_{l=1}^L$ to represent the continuous function, where the function value ${\bf H}({\bf v})$ of each input coordinate $\bf v$ is obtained by linear interpolation over the grids.} In practice, the grid resolution $N_l$ varies for different $l$, leading to a powerful multi-resolution representation for complex function modeling \cite{muller2022instant}. The encoding vector ${\bf H}({\bf v})$ is then passed through a lightweight MLP to produce the desired function value at $\bf v$ for continuous representation \cite{muller2022instant}. Each grid tensor $\mathcal{G}_{l}$ is efficiently stored in a hash table with length $T$ and width $F$. Specifically, the element of $\mathcal{G}_{l}$ is queried by a predefined hash function \cite{muller2022instant} to connect the tensor index and the hash table index, which establishes one-to-one or one-to-many (in the case of $T$ being small) correspondence between the hash table and the grid tensor $\mathcal{G}_{l}$. The hash table enhances the efficiency of multi-resolution grid encoding for continuous representation. However, as the dimensionality $D$ increases, InstantNGP encounters {\bf significant storage requirements and computational expenses} due to the multi-resolution {\bf high-dimensional grids $\{\mathcal{G}_{l}\}_{l=1}^L$ in ${\bf H}({\bf v})$}. The parameter count in InstantNGP is $O(LFN_l^D)$, which grows exponentially w.r.t. the data dimension $D$. For typical CI reconstruction problems, such as compressive MRI and SCI, the requirements for high-dimensional image representation makes InstantNGP less efficient. Therefore, we propose the tensor decomposed multi-resolution grid encoding model to release the storage and computational burden for high-dimensional image representation.
\subsubsection{Motivation and Formulation of GridTD}\label{subsubsec:mot}
The main motivation of the proposed method is to decompose the dense multi-resolution grids of InstantNGP into lightweight one-dimensional (1D) grids in terms of tensor decomposition, which alleviates parameter redundancy for high-dimensional data. Especially, InstantNGP interpolates between the feature vectors stored in heavy $D$-dimensional grids. To address the heavy storage costs, we first decouple the $D$-dimensional coordinate $\bf v$ into $D$ independent 1D coordinates ${\bf v}[d]$ ($d=1,\cdots,D$). Each function value of the 1D coordinates is interpolated and encoded through {\bf lightweight 1D multi-resolution grids}, i.e., ${\bf H}_d(\mathbf{v}[d])$ ($d=1,\cdots,D$). Then we employ the element-wise product to fuse these 1D encodings into a unified feature $\mathbf{H}_{\rm GridTD}(\mathbf{v})$ (inspired by the CP tensor decomposition, as we will discuss in Lemma \ref{def:gtdp}). For simplicity, we let $R:=LF$ be the CP rank of the GridTD model.
\begin{definition}[Tensor decomposed multi-resolution grid encoding (GridTD)]\label{def:gtd}
Given a normalized input coordinate vector $\mathbf{v} \in [0,1)^D$, the tensor decomposed multi-resolution grid encoding $\mathbf{H}_{\rm GridTD}(\mathbf{v}):[0,1)^{D}\rightarrow{\mathbb R}^{R}$ is defined as:
\begin{equation}\small\label{eq_gridtd}
    \mathbf{H}_{\rm GridTD}(\mathbf{v}) :=\left(\bigodot_{d=1}^D {\bf H}_d(\mathbf{v}[d])\right)\in{\mathbb R}^{R},
\end{equation}
where $R:=LF$ is the CP rank, $\odot$ is the element-wise product, and each component (or say factor vector) ${\bf H}_d(\mathbf{v}[d]):[0,1)\rightarrow{\mathbb R}^{R}$ ($d=1,\cdots,D$) is a one-dimensional multi-resolution grid encoding defined in Definition \ref{def:mrh}.
\end{definition}
Similar to InstantNGP, the encoding feature $\mathbf{H}_{\rm GridTD}(\mathbf{v})$ is processed through an MLP to reconstruct the target value. {Compared to InstantNGP, the GridTD solely uses lightweight 1D multi-resolution grids for the $D$-dimensional function through coordinate decomposition, which can be stored in smaller hash tables and hence significantly eases storage costs (see Table \ref{tab:complexity}).}
Note that in Definition \ref{def:gtd} we have defined the GridTD as a continuous function w.r.t. the given coordinate vector $\bf v$. We can use the formulation to process multiple coordinates in parallel and write the output of the GridTD as a tensor. Given $n_1n_2\cdots n_D$ coordinates of a $D$-th order tensor ${\cal X}\in{\mathbb R}^{n_1\times\cdots\times n_D}$, we can decompose these coordinates into $n_1+n_2+\cdots+n_D$ separated coordinates that represent the individual coordinates of each dimension\footnote{For example, for a 2D grid $(1,1),(1,2),(2,1),(2,2)$, we can use two spatial coordinate vectors ${\bf v}_1=(1,2)$ and ${\bf v}_2=(1,2)$ to fully determine the 2D grid by Cartesian product between ${\bf v}_1$ and ${\bf v}_2$.}. We can then write the tensor parallelism formulation of GridTD for these decomposed coordinates. The tensor parallelism facilitates efficient parallel processing of tensor coordinates, making GridTD more computationally efficient than InstantNGP for high-dimensional image representation in CI reconstruction.   
\begin{lemma}[GridTD under tensor parallelism]\label{def:gtdp}
Given a tensor of size $n_1\times n_2\times \cdots\times n_D$ and its decomposed coordinate vectors ${\bf v}_d\in{[0,1)}^{n_d}$ ($d=1,\cdots,D$)\footnote{For instance, we can set ${\bf v}_d=(0,1/{n_d},2/{n_d},\cdots,(n_d-1)/n_d)$ to be the uniformly sampled coordinates of a tensor along the $d$-th dimension.}, we define ${\cal H}\in{\mathbb R}^{n_1\times \cdots\times n_D\times R}$ as the GridTD encoding of this tensor:
$$
{\cal H}[i_1,i_2,\cdots,i_D, :]:=\mathbf{H}_{\rm GridTD}({\bf v}_1[i_1],{\bf v}_2[i_2],\cdots,{\bf v}_D[i_D]).$$
Then ${\cal H}$ can be written as
\begin{equation}\small\label{eq:tensor}
\begin{split}
&{\cal H}=\left(\bigotimes_{d=1}^D {\bf h}_{1,d},\bigotimes_{d=1}^D {\bf h}_{2,d},\cdots,\bigotimes_{d=1}^D {\bf h}_{R,d}\right),
\end{split}
\end{equation}
where the concatenation is performed in the last dimension, ${\bf h}_{r,d}:=\left(\bigoplus_{i=1}^{n_d}{\bf H}_d({\bf v}_d[i])[r]\right)\in{\mathbb R}^{n_d}$ is the factor vector of the encoding tensor $\cal H$, and ${\bf H}_d(\cdot):[0,1)\rightarrow{\mathbb R}^{R}$ is the 1D multi-resolution grid encoding defined in Definition \ref{def:gtd}.
\end{lemma}
\begin{proof}
The result is a direct repetitive calculation of \eqref{eq_gridtd} for the inferred coordinates.
\end{proof}
Lemma \ref{def:gtdp} indicates that GridTD can encode tensors in a compact form \eqref{eq:tensor}, resulting in efficient computation (see Table \ref{tab:complexity}). If we perform summing over ${\cal H}$ along the last dimension in \eqref{eq:tensor}, then we recover the CP decomposition formulation $\sum_{r=1}^R\bigotimes_{d=1}^D {\bf h}_{r,d}$ as introduced in \eqref{CP}. Hence, our GridTD is essentially related to this tensor decomposition technique from its encoding structure. Differently, while CP decomposition directly sums the rank-1 components $\bigotimes_{d=1}^D {\bf h}_{r,d}$, we will instead use a lightweight MLP $g_\Theta:{\mathbb R}^R\to{\mathbb R}$ to modulate the rank fusion process by operating on the rank dimension $R$ of the encoding tensor $\cal H$. An illustration of the tensor parallelism formulation of GridTD for $D=3$ is shown in Fig. \ref{fig:gradient}.
\subsubsection{The Efficiency of GridTD}\label{subsubsec:effi}
As compared with InstantNGP, GridTD holds favorable storage and computational efficiency by the decomposition structure. First, for storage optimization, GridTD requires only $D$ one-dimensional grids compared to InstantNGP's $D$-dimensional grids. The storage complexity (parameter count) of InstantNGP is $\underline{O(LFN_l^D)}$, while the storage complexity of GridTD is $\underline{O(LFDN_l)}$. GridTD significantly reduces the number of parameters required to represent a tensor (see Table \ref{tab:complexity} for examples).\par 
\begin{table}[t]
\scriptsize
\tabcolsep = 2pt
\centering
\caption{Storage complexity, parameter number, computational complexity, and running time (300 iterations for inpainting) for a $D$-th order tensor of size $n\times\cdots\times n$ ($n=100, D=3$ here). $L$ denotes the resolution number, $F$ denotes the feature vector length, and $N_l$ denotes the grid resolution.}
\label{tab:complexity}
\begin{tabular}{lcccc}
\toprule
Method & Storage complexity & \#Params.&Computational complexity & Time \\
\midrule
InstantNGP & $\mathcal{O}(LFN_l^D)$ &3.12M& $\mathcal{O}((2n)^DLF)$ & 38.05s\\
GridTD & $\mathcal{O}(LFDN_l)$ &0.02M& $\mathcal{O}(2nDLF)$ & 1.31s\\
\bottomrule\vspace{-0.5cm}
\end{tabular}
\end{table}
Second, GridTD significantly reduces the forward computational complexity for tensor data. From Lemma \ref{def:gtdp}, we can see that to generate the encoding tensor $\cal H$, each 1D grid function ${\bf H}_d(\cdot)$ is queried only $n_d$ times because the 1D encoding ${\bf H}_d({\bf v}_d[i])$ can be shared in parallel for all rank components $r=1,2,\cdots,R$. Hence, GridTD solely performs $\sum_{d=1}^D n_d$ times of lightweight one-dimensional grid encoding to generate the large encoding tensor $\cal H$, and the computational complexity of GridTD is $\underline{O(2nDLF)}$ (the total number of linear interpolations over 1D grids), where $n=\max_dn_d$. As compared, the original InstantNGP \cite{muller2022instant} would require performing $\prod_{d=1}^Dn_D$ times of $D$-dimensional grid encoding to generate the same-sized encoding tensor. The computational complexity of InstantNGP is $\underline{O((2n)^DLF)}$ (the total number of linear interpolations over $D$-dimensional grids) for the same-sized tensor, which is much larger than that of GridTD. The computational efficiency of GridTD makes it highly efficient for image representation, requiring less time for optimization (see Table \ref{tab:complexity} for examples). The storage and computational advantages enable GridTD to outperform InstantNGP in both memory efficiency and optimization speed. Moreover, GridTD maintains equivalent hierarchical modeling abilities with InstantNGP, since 1) GridTD also uses multi-resolution grids for hierarchical function modeling, and 2) for any high-dimensional tensor, it can be exactly factorized into the CP decomposition with proper CP rank $R$ \cite{SIAM_review}, confirming the unblemished representation capacity of GridTD.
\subsubsection{The Effectiveness of GridTD}\label{subsubsec:effe}
From a data modeling perspective, GridTD inherently incorporates two structural modeling techniques. First, GridTD efficiently captures complex hierarchical structures of images from the tensor decomposed multi-resolution grid encoding \cite{SIAM_review}. Second, GridTD captures the inherent smoothness of high-dimensional images from the continuous function representation of parametric grid encoding. More importantly, GridTD enhances the robustness of such smooth characterization across different data dimensions $D$. Especially, instead of performing the $D$-dimensional interpolation directly, GridTD decomposes the problem into one-dimensional factors and applies 1D interpolation separately to each 1D factor. This makes its smoothness characterization stable and less sensitive to dimensionality $D$. To verify this, we analyze the Lipschitz continuity of GridTD and InstantNGP, demonstrating the key theoretical advantage of GridTD: its Lipschitz constant is less sensitive to the input dimension $D$, whereas InstantNGP’s Lipschitz constant scales exponentially with $D$. This property ensures stable and consistent performance for GridTD regardless of the input dimensionality, making GridTD particularly suitable for high-dimensional image representation in CI reconstruction.
\begin{theorem}[Lipschitz smooth bound for InstantNGP]\label{theorem-1} Consider the composition of an MLP $g_\Theta(\cdot):{\mathbb R}^{LF}\to{\mathbb R}$ and the multi-resolution grid encoding $\mathbf{H}(\cdot):[0,1)^D\to{\mathbb R}^{LF}$ defined in Definition \ref{def:mrh}: 
\begin{equation}\small\begin{split}
	(g_\Theta\circ{\bf H})({\bf v}):= {\bf W}_2\sigma\left({\bf W}_1\mathbf{H}(\mathbf{v})+\mathbf{b}\right):[0,1)^D\to{\mathbb R},
	\end{split}
\end{equation}	
where ${\bf W}_1\in{\mathbb R}^{n_1\times LF},{\bf W}_2\in{\mathbb R}^{1\times n_1},\mathbf{b} \in {\mathbb R}^{n_1\times 1}$ are weights and bias of the MLP and $\sigma$ is a $\gamma$-Lipschitz smooth activation function. Assume that the grid norm $\lVert \mathcal{G}_l[\mathbf{z},:]\rVert_{\ell_1}\leq1$ for all resolutions $l$ and vertices ${\bf z}\in{\mathbb Z}^D$. Then for any coordinates ${\bf v}_1,{\bf v}_2\in[0,1)^{D}$, the following Lipschitz smooth bound holds:
\begin{equation}\small
\left|(g_\Theta\circ{\bf H})({\bf v}_1)-(g_\Theta\circ{\bf H})({\bf v}_2)\right|\leq
2^D \gamma \eta D N \lVert{\bf v}_1-{\bf v}_2\rVert_{\ell_1},
\end{equation}	
where $ \eta=\lVert{\bf W}_1\rVert_{\ell_1}\lVert{\bf W}_2\rVert_{\ell_1}$ and $N=\underset{l=1}{\overset{L}{\sum}}(N_l-1)$.\\
\end{theorem}
\begin{theorem}[Lipschitz smooth bound for GridTD]\label{theorem-2} Consider the composition of an MLP $g_\Theta(\cdot):{\mathbb R}^{R}\to{\mathbb R}$ and the tensor decomposed multi-resolution grid encoding $\mathbf{H}_{\rm GridTD}(\cdot):[0,1)^D\to{\mathbb R}^{R}$ defined in Definition \ref{def:gtd}: 
\begin{equation}\small
\label{gtdmlp}
\begin{split}
	(g_\Theta\circ{\bf H}_{\rm GridTD})({\bf v}):= {\bf W}_2\sigma\left({\bf W}_1\mathbf{H}_{\rm GridTD}(\mathbf{v})+\mathbf{b}\right).
	\end{split}
\end{equation}	
Let the assumptions in Theorem \ref{theorem-1} hold. Then for any coordinates ${\bf v}_1,{\bf v}_2\in[0,1)^{D}$, the following Lipschitz smooth bound holds for GridTD:
 \begin{equation}\small
\begin{split}
&|(g_\Theta\circ{\bf H}_{\rm GridTD})({\bf v}_1)-(g_\Theta\circ{\bf H}_{\rm GridTD})({\bf v}_2)|\\\leq&
2\gamma \eta DN \lVert{\bf v}_1-{\bf v}_2\rVert_{\ell_1},
\end{split}
\end{equation}
  where $\eta=\lVert{\bf W}_1\rVert_{\ell_1}\lVert{\bf W}_2\rVert_{\ell_1}$ and $N=\underset{l=1}{\overset{L}{\sum}}(N_l-1)$. \\
\end{theorem}
We observe that the Lipschitz smooth bound of GridTD scales linearly with dimension $D$, whereas the Lipschitz bound of InstantNGP scales exponentially with $D$. Consequently, under the same configuration (such as weight initialization), the degree of smoothness encoded in GridTD does not change abruptly across different dimensions, while the smoothness of InstantNGP would be more sensitive to dimensional changes. Building on this insight, we derive generalization error bounds for both methods. Since GridTD’s Lipschitz constant remains stable w.r.t. dimension, its generalization error scales more favorably compared to InstantNGP, particularly in high-dimensional settings. This property makes GridTD more robust and effective for high-dimensional image representation in compressive sensing, as its generalization error bound is less sensitive to the image dimension.
\begin{theorem}[Generalization error bound via Lipschitz smoothness]\label{theorem-3}
Suppose the assumption in Theorem \ref{theorem-2} holds. Let $\mathcal{X} \times \mathcal{Y} \subset \left[ 0,1 \right)^D \times \left[ 0,1 \right]$ (coordinate $\times$ image value), and ${\mathscr D}$ is a data distribution over $\mathcal{X} \times \mathcal{Y}$. Consider the GridTD model $(g_\Theta\circ{\bf H}_{\rm GridTD}): \left[ 0,1 \right)^D \to \left[ 0,1 \right]$ with Lipschitz constant $L = \max_{{\bf x}} \|\nabla (g_\Theta\circ{\bf H}_{\rm GridTD})({\bf x})\|\leq2\gamma\eta DN$. Then with probability $1 - \delta$ over the sample of training data $\{({\bf x}_1,y_1), \ldots, ({\bf x}_n,y_n)\} \sim {\mathscr D}$, we have the following generalization error bound for $\mathcal{G}[g_\Theta\circ{\bf H}_{\rm GridTD}] := \mathbb{E}_{({\bf x},y) \sim {\mathscr D}} ((g_\Theta\circ{\bf H}_{\rm GridTD})({\bf x}) - y)^2 - \frac{1}{n} \sum_{i=1}^n ((g_\Theta\circ{\bf H}_{\rm GridTD})({\bf x}_i) - y_i)^2:$
\begin{equation}
  \mathcal{G}[g_\Theta\circ{\bf H}_{\rm GridTD}] \leq \frac{8\gamma\eta DN}{\sqrt{n}} + 3 \sqrt{\frac{\log \frac{2}{\delta}}{2n}}.  
\end{equation}
\small

\end{theorem}
Similar to Theorem \ref{theorem-3}, we can deduce the generalization error bound of the InstantNGP model $(g_\Theta\circ{\bf H}): \left[ 0,1 \right)^D \to \left[ 0,1 \right]$ from its Lipschitz bound $2^D\gamma\eta DN:$
\[\small
\mathcal{G}[g_\Theta\circ{\bf H}] \leq \frac{2^{D+2}\gamma\eta DN}{\sqrt{n}} + 3 \sqrt{\frac{\log \frac{2}{\delta}}{2n}}.
\]
We observe that the error bound of InstantNGP scales exponentially with dimension $D$, while the error bound of GridTD scales linearly due to their respective Lipschitz bounds. This suggests GridTD maintains stable performance and generalization across dimensions, whereas InstantNGP may struggle.\par 
\begin{table}[t]
\tabcolsep=4pt
\centering
\caption{PSNR comparison of InstantNGP (INGP) and GridTD for data inpainting with different data dimensions $D$.}
\label{tab:gen}
\begin{tabular}{ccccccc}
\hline
\multirow{2}{*}{Dimension} & \multicolumn{2}{c}{SR=0.2}  & \multicolumn{2}{c}{SR=0.1}& \multicolumn{2}{c}{SR=0.05}  \\
\cline{2-7}
 &  INGP &GridTD \;& INGP &GridTD  \;& INGP &GridTD  \\
\hline
$D=1$ &  32.54&32.54\;& 26.93&26.93\;&  25.62&25.62 \\
$D=2$ & 29.24&\bf34.38\;&26.24& \bf29.81\;&  23.85 &\bf26.50 \\
$D=3$ & 27.75&\bf37.09\;& 24.30& \bf36.87\;& 21.87&\bf36.22   \\
\hline
\end{tabular}\vspace{-0.2cm}
\end{table}
\begin{figure}[t]
\scriptsize
    \centering
    \setlength{\tabcolsep}{2pt} 
    \begin{tabular}{ccccc}
        \multirow{1}*{\rotatebox[origin=c]{90}{$D=1$\hspace{-1.2cm}}}&\includegraphics[width=0.105\textwidth]{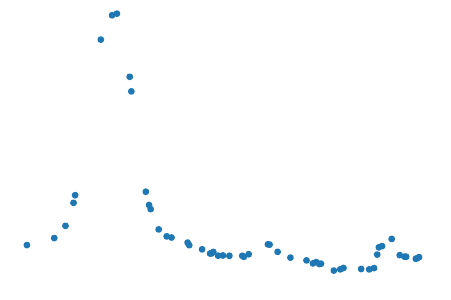} &
        \includegraphics[width=0.105\textwidth]{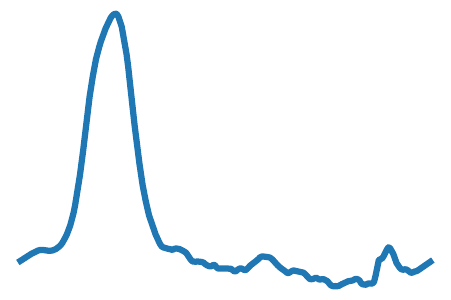} &
        \includegraphics[width=0.105\textwidth]{figs/Toy/output1D.pdf} &
        \includegraphics[width=0.105\textwidth]{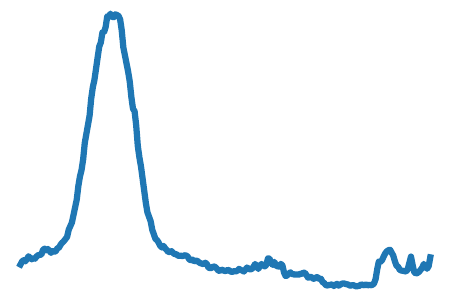} \\
        \multirow{1}*{\rotatebox[origin=c]{90}{$D=2$\hspace{-1.5cm}}}&\includegraphics[width=0.105\textwidth]{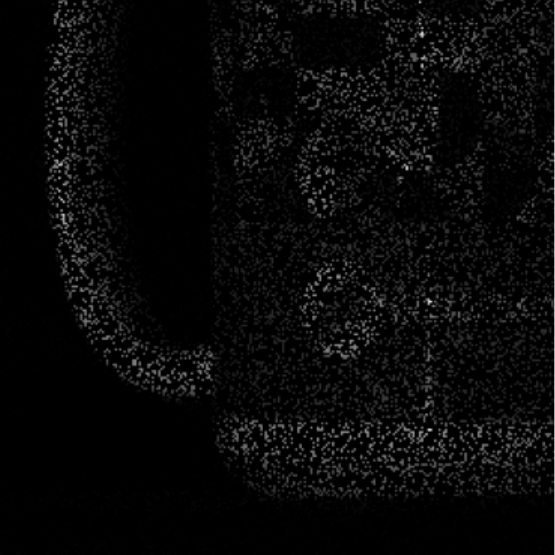} &
        \includegraphics[width=0.105\textwidth]{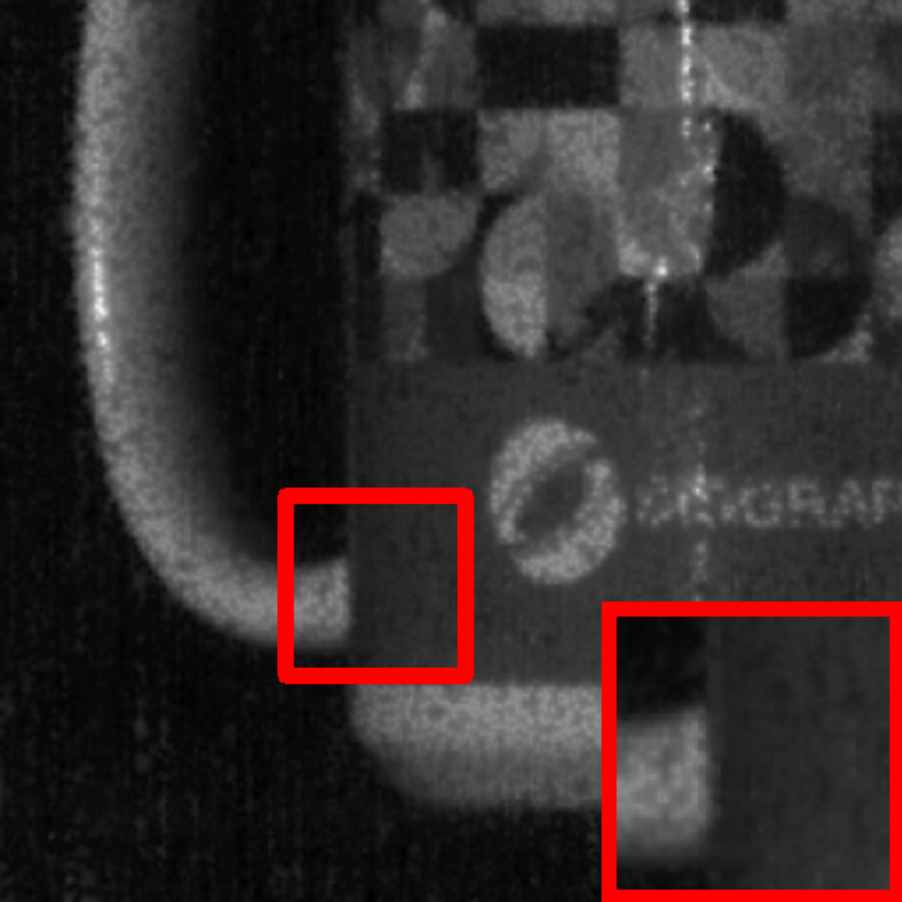} &
        \includegraphics[width=0.105\textwidth]{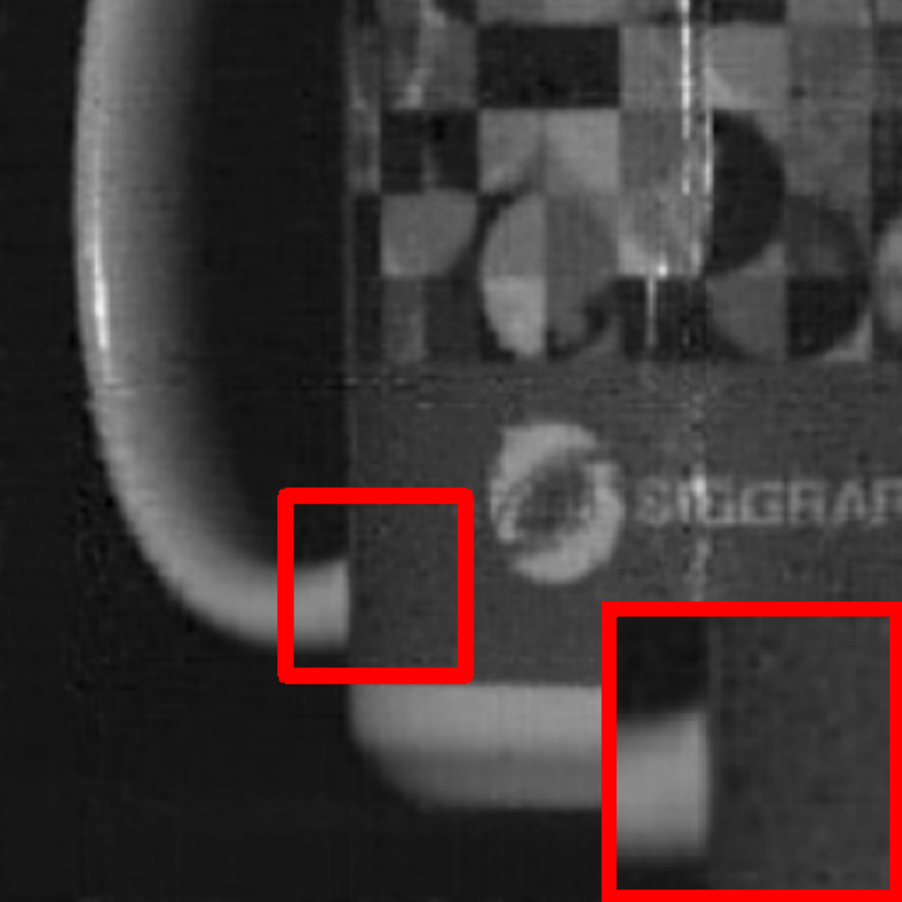} &
        \includegraphics[width=0.105\textwidth]{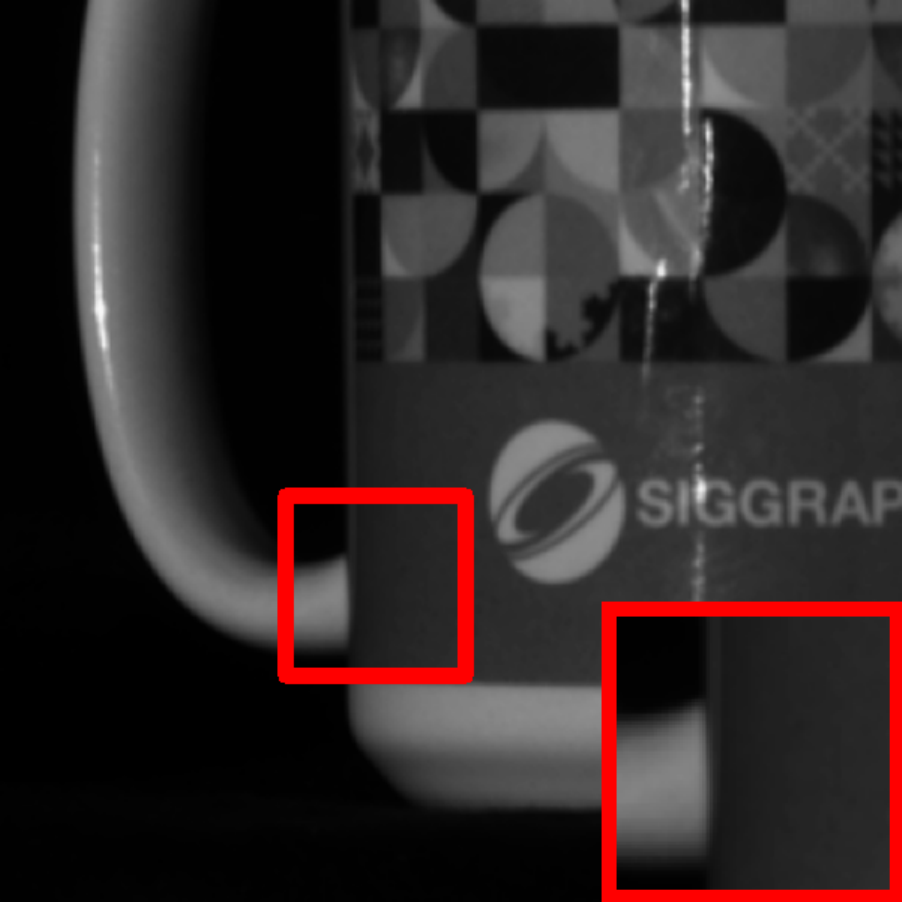} \\
        \multirow{1}*{\rotatebox[origin=c]{90}{$D=3$\hspace{-1.5cm}}}&\includegraphics[width=0.105\textwidth]{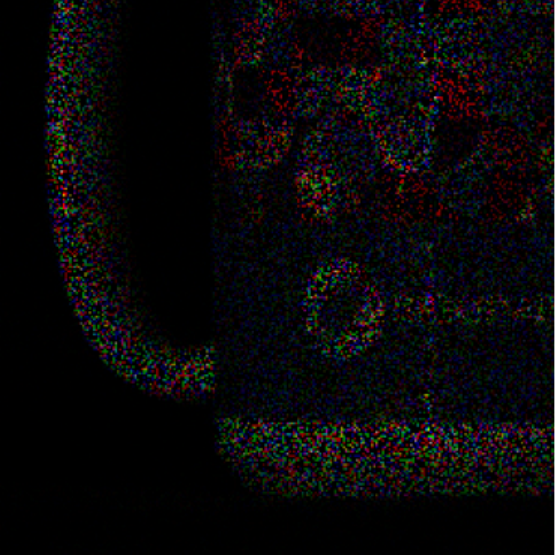} &
        \includegraphics[width=0.105\textwidth]{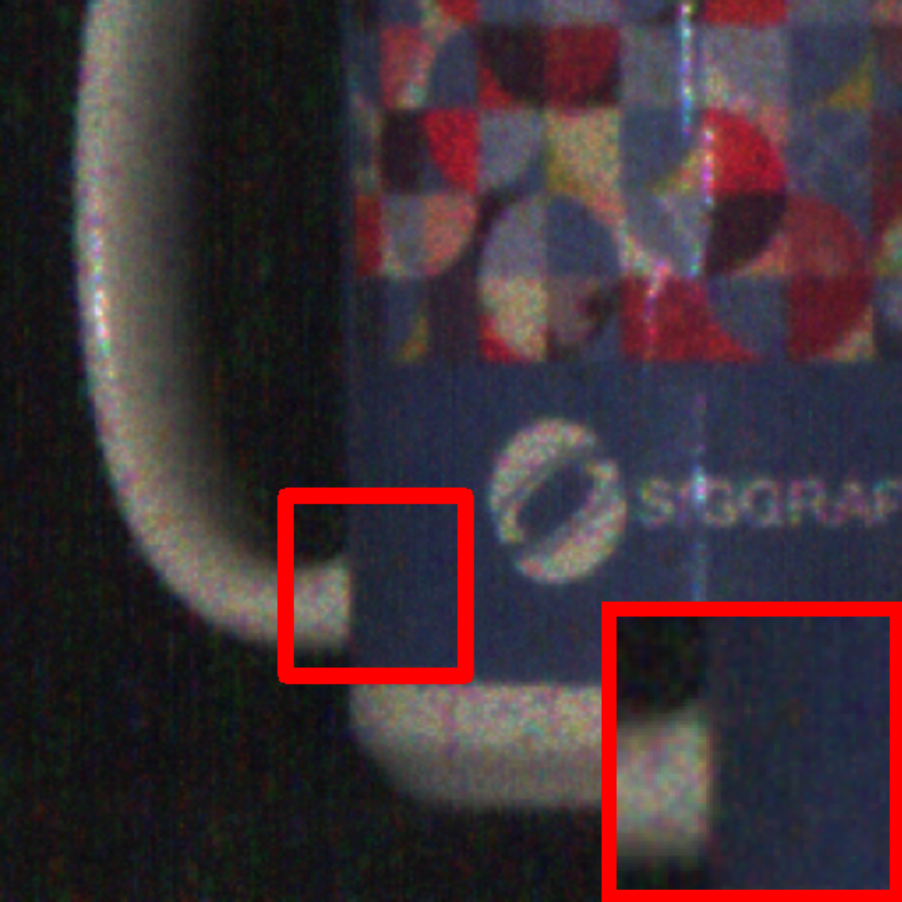} &
        \includegraphics[width=0.105\textwidth]{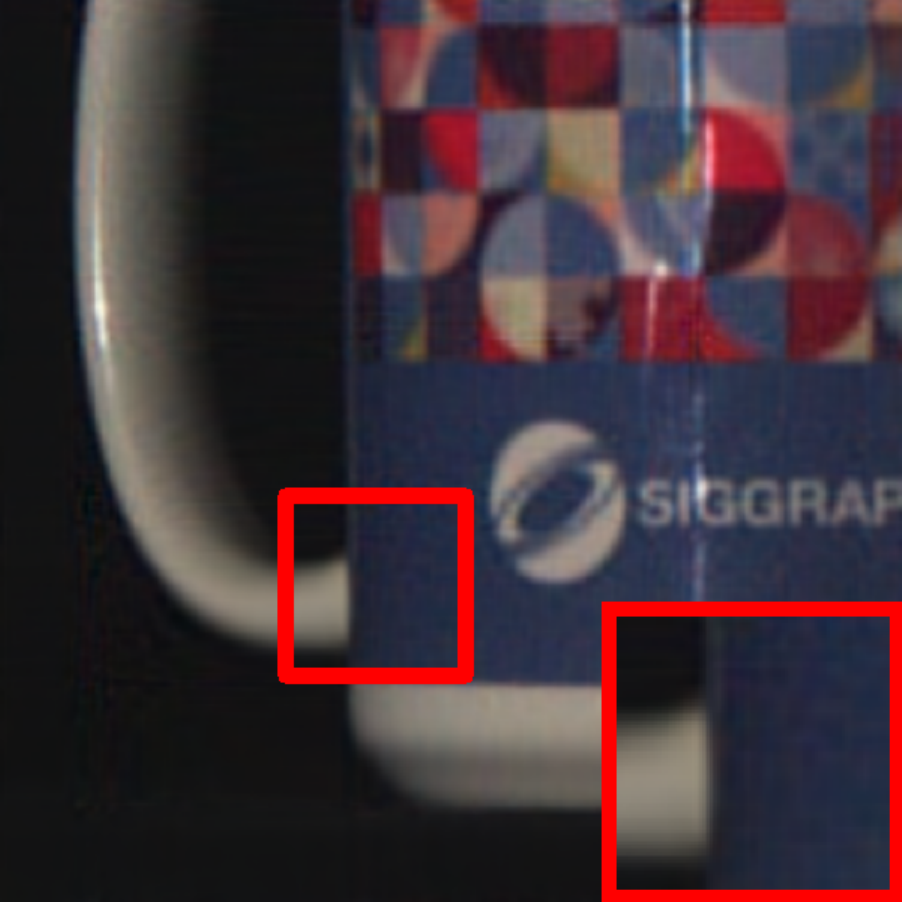} &
        \includegraphics[width=0.105\textwidth]{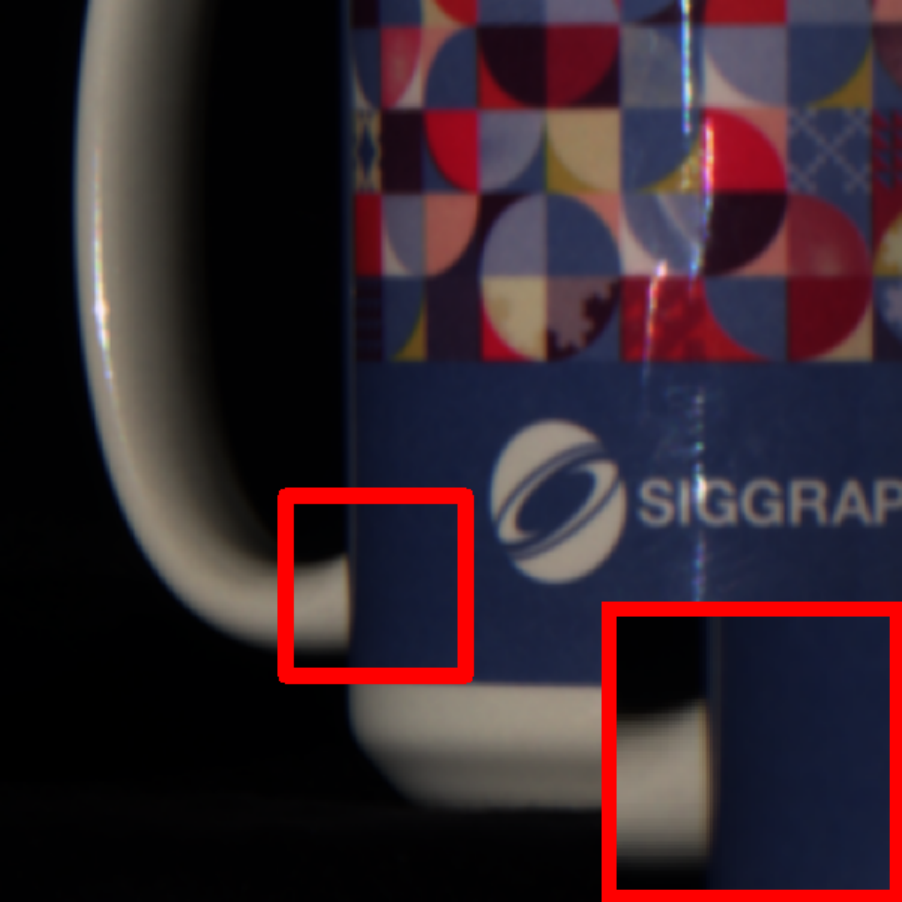} \\
    &Measurement & InstantNGP & GridTD & Original\\
    \end{tabular}
    \caption{Toy examples on the data inpainting task. From top to down list the inpainting results for 1D, 2D, and 3D signals using InstantNGP and our GridTD under the same hyperparameter configuration. GridTD enjoys favorable generalization abilities across different data dimensions by virtue of its more stable generalization error bound w.r.t. the dimension $D$.}\vspace{-0.4cm}
    \label{fig:generalize}
\end{figure}
To verify the theoretical analysis, we conduct toy experiments on the inpainting task to comparing how well GridTD and InstantNGP generalize across different dimensions (vectors, matrices, and tensors). All experimental settings, including grid tensor initialization strategies, are kept consistent across dimensions to test the robustness of different methods. The inpainting results under different data dimensions (vectors, matrices, and tensors) and different sampling rates (SRs) (the proportion of observed entries) are shown in Table \ref{tab:gen} and Fig. \ref{fig:generalize}. We can observe that while GridTD and InstantNGP perform equally well for the one-dimensional example\footnote{Indeed, GridTD recovers to the InstantNGP model when $D=1$.}, the performance of InstantNGP degrades when increasing dimensions $D=2,3$. As compared, GridTD remains stable for higher dimensions $D=2,3$. The results align with our theory that GridTD holds more stable Lipschitz smoothness and generalization error bounds when $D$ changes (Theorems \ref{theorem-2}-\ref{theorem-3}), making it more robust across dimensions. In practice, InstantNGP would require reconfigurations of the hyperparameter (such as the variance of Gaussian initialization) to adapt to different data dimensions, while GridTD is more stable and versatile for different dimensions, demonstrating GridTD’s stronger generalization ability across dimensions.
\subsection{Temporal Affine Adapter and Smooth Regularization}\label{subsec:Aff}
To improve the dynamic structure modeling of GridTD for dynamic high-dimensional data (such as temporal videos), we further propose two key technical improvements when adapting our framework to compressive imaging reconstruction.\par 
First, we introduce the temporal affine adapter, which adaptively learns temporal motion patterns (scaling, rotation, and translation) between video frames by reusing the temporal grid encoding of GridTD. Traditional tensor decomposition methods rely on the assumption that video data exhibits low-rank structures along its temporal dimension. However, complex background interference often reduces frame-to-frame correlations along the temporal dimension. This leads to deviations from an ideal low-rank tensor structure. To address this issue, we adopt a relaxed pseudo low-rank assumption, which assumes that the original video tensor can be derived from a latent low-rank tensor through frame-wise affine transformations. We then learn the affine transformation parameters adaptively during optimization. Let the generated low-rank tensor of GridTD be $\mathcal{L} \in \mathbb{R}^{n_1 \times n_2 \times n_3}$, where $n_1$ and $n_2$ denote the spatial dimensions, and $n_3$ is the number of frames. The temporal affine adapter $\mathcal{A}:\mathbb{R}^{n_1 \times n_2 \times n_3}\to\mathbb{R}^{n_1 \times n_2 \times n_3}$ applies a frame-wise affine transformation $\Lambda^{(t)}$ to each frame $\mathcal{L}[:, :, t]$ followed by bilinear interpolation, defined as
\begin{equation}\small\label{affine}
  \mathcal{A}(\mathcal{L})[i, j, t] := \mathrm{Bilinear}\left(\mathcal{L}[:, :, t],\, \Lambda^{(t)}(i, j)\right).  
\end{equation}
Here, $\mathrm{Bilinear}(\cdot, \cdot)$ denotes the bilinear interpolation operation, and $\Lambda^{(t)}(i, j)\in{\mathbb R}^2$ gives the affined position of the coordinate $(i, j)$ after the affine transformation in frame $t$. The affined tensor $\mathcal{A}(\mathcal{L})\in\mathbb{R}^{n_1 \times n_2 \times n_3}$, derived from the latent low-rank tensor $\cal L$, is designed to more accurately capture the dynamic structure between video frames. The affine transformation $\Lambda^{(t)}$ here is defined as
\[
\Lambda^{(t)}(i, j) :=
\begin{pmatrix}
s^{(t)} \cos\theta^{(t)} & -s^{(t)} \sin\theta^{(t)} & b_x^{(t)} \\
s^{(t)} \sin\theta^{(t)} & \;\;s^{(t)} \cos\theta^{(t)} & b_y^{(t)}
\end{pmatrix}
\begin{pmatrix}
i \\
j \\
1
\end{pmatrix},
\]
where the scaling parameter $s^{(t)}$ and the rotation angle $\theta^{(t)}$ are learnable parameters for each frame, directly optimized during training. However, we found that directly setting the translation parameters $b_x^{(t)},b_y^{(t)}$ as learnable parameters fails to capture the temporal translation dependencies across frames. Hence, we propose the temporal affine adapter by introducing additional INRs to continuously learn translation parameters $b_x^{(t)},b_y^{(t)}$, by reusing the temporal grid encoding ${\bf H}_3(\cdot):[0,1)\to{\mathbb R}^R$ of GridTD (defined in Definition \ref{def:gtd}). This effectively leverages the temporal feature of GridTD to enhance the robustness of the affine transformation. Specifically, the translation parameters $(b_x^{(t)}, b_y^{(t)})$ are generated by two INRs $f_x(\cdot),f_y(\cdot):{\mathbb R}^R\to {\mathbb R}$ conditioned on the temporal grid encoding ${\bf H}_3(\cdot)$:
\[
b_x^{(t)} = f_x\left(\mathbf{H}_3\left(\frac{t-1}{n_3}\right)\right),\; b_y^{(t)} = f_y\left(\mathbf{H}_3\left(\frac{t-1}{n_3}\right)\right),
\]
where $\mathbf{H}_3(\frac{t-1}{n_3})$ denotes the encoded temporal feature of frame $t$, and $f_x(\cdot)$ and $f_y(\cdot)$ are two INRs with shared inputs but independent parameters. The INRs are optimized along with GridTD during unsupervised reconstruction. The temporal affine adapter introduced above utilizes the conditional dependency and intrinsic correlation along the temporal dimension by reusing the temporal grid encoding $\mathbf{H}_3(\cdot)$, hence enabling better temporal coherence across frames. In experiments, the temporal affine adapter is used for video SCI reconstruction.\par
Furthermore, to better exploit the spatial-spectral or spatial-temporal smoothness of the data, we further leverage a smooth regularization term using the spatial-spectral total variation (SSTV). The spatial TV for an image ${\cal X}\in{\mathbb R}^{n_1\times n_2\times n_3}$ is defined as $\mathrm{TV}({\cal X}) = \| {\bf D}_x {\cal X}  \|_{\ell_1} + \| {\bf D}_y{\cal X} \|_{\ell_1}$, where ${\bf D}_x {\cal X}= {\cal X}[2:n_1,:,:] - {\cal X}[1:n_1-1,:,:]$ and ${\bf D}_y {\cal X} = {\cal X}[:,2:n_2,:] - {\cal X}[:,1:n_2-1,:]$ are discrete difference tensors. The second-order SSTV for an image ${\cal X}$ is defined as
\begin{equation}\small
\label{SSTV}
\mathrm{SSTV}({\cal X}) = \| {\bf D}_x ({\bf D}_z {\cal X})  \|_{\ell_1} + \| {\bf D}_y({\bf D}_z {\cal X}) \|_{\ell_1},
\end{equation}
where ${\bf D}_z{\cal X} = {\cal X}[:,:,2:n_3] - {\cal X}[:,:,1:n_3-1]$ is the spectral/temporal difference tensor. The SSTV captures the second-order spatial-temporal or spatial-spectral smoothness of high-dimensional data to improve robustness against noise \cite{sstv}. We impose the spatial TV and SSTV on the recovered tensor data to further enhance such robustness.
\subsection{GridTD-Based ADMM Algorithm for CI Reconstruction}\label{subsec:Alg}
For compressive imaging reconstruction, we integrate the GridTD representation model with the plug-and-play alternating direction method of multipliers (PnP-ADMM) algorithmic framework \cite{PnP}. For simplicity, we take the video SCI \eqref{eq:sci_model_alt} as an example to detail the PnP-ADMM algorithm, and the PnP-ADMM for dynamic MRI and spectral SCI can be deduced analogously.
\begin{algorithm}[t]
    \renewcommand\arraystretch{1.1}
    \caption{{GridTD-based plug-and-play ADMM algorithm for video snapshot compressive imaging.}}\label{alg:GridTD}
    \begin{algorithmic}[1]
        \renewcommand{\algorithmicrequire}{\textbf{Input:}} 
        \REQUIRE Compressed measurement $\mathbf{Y}$, sensing masks $\{\mathcal{M}_t\}$, maximum iteration $K$, $\kappa>1$;
        \renewcommand{\algorithmicrequire}{\textbf{Initialization:}} 
        \REQUIRE Initialize $\mathcal{X}^1 = 0$, $\mathcal{V}^1 = 0$, $\mathcal{U}^1 = 0$, randomly initialize GridTD parameters $\Theta$;
        \FOR{$k = 1$ to $K$}
             \STATE\textbf{Update $\mathcal{X}^{k+1}$} via closed-form solution \eqref{eq:close_form};
            \STATE\textbf{Update ${\cal V}^{k+1}$} via optimizing the GridTD parameters $\Theta$ according to \eqref{eq:v_subproblem_tensor} using  backpropagation;
             \STATE {\bf Update Lagrange multiplier} via $\mathcal{U}^{k+1} = \mathcal{U}^k + \mathcal{X}^{k+1} - \mathcal{V}^{k+1}$, $\rho^{k+1} = \kappa \rho^k$;
        \ENDFOR
        \renewcommand{\algorithmicrequire}{\textbf{Output:}} 
        \REQUIRE The reconstruction ${\cal X}$.
    \end{algorithmic}
\end{algorithm}
By leveraging the proposed GridTD as the high-dimensional image representation, the optimization model for the video SCI problem \eqref{eq:sci_model_alt} is formulated as
\begin{equation}\small\small
\min_{{\cal X},{\cal V}} \frac{1}{2}\left\|{\bf Y}-\sum_{t=1}^{n_3} \mathcal{M}_t \odot \mathcal{X}_t\right\|_F^2+{\rm GridTD}({\cal V}),\;{\rm s.t.}\;{\cal X}={\cal V},
\label{eq:pnp_admm}
\end{equation}
where $\mathbf{Y} \in \mathbb{R}^{n_1 \times n_2}$ denotes the compressed measurement, $\mathcal{X} \in \mathbb{R}^{n_1 \times n_2 \times n_3}$ denotes the desired reconstruction, $\mathcal{M}_t$ denotes the $t$-th mask corresponding to the $t$-th frame $\mathcal{X}_t$, $\mathcal{V}$ is an auxiliary variable, and ${\rm GridTD}({\cal V})$ denotes an implicit regularizer associated with the GridTD model. The constrained problem can be decomposed into subproblems via ADMM:
\begin{equation}\small\label{admm_eqs}
\small
\left\{
\begin{aligned}
\mathcal{X}^{k+1} &= \arg\min_{\mathcal{X}} \ \frac{1}{2} \left\| \mathbf{Y} - \sum_{t=1}^{n_3} \mathcal{M}_t \odot \mathcal{X}_t \right\|_F^2 + \frac{\rho^k}{2} \left\| \mathcal{X} - \mathcal{V}^k + \mathcal{U}^k \right\|_F^2, \\
\mathcal{V}^{k+1} &= \arg\min_{\mathcal{V}} \  \text{GridTD}(\mathcal{V}) + \frac{\rho^k}{2} \left\| \mathcal{V} - \mathcal{X}^{k+1} - \mathcal{U}^k \right\|_F^2, \\
\mathcal{U}^{k+1} &= \mathcal{U}^k + \mathcal{X}^{k+1} - \mathcal{V}^{k+1},\;\rho^{k+1} = \kappa \rho^k,
\end{aligned}
\right. 
\end{equation}
where $\mathcal{U}$ is the Lagrange multiplier, $\rho$ is a penalty parameter, and $\kappa>1$ is a constant. For each frame $t$, the exact solution of the $\mathcal{X}$-subproblem is:
  \begin{equation}\small
  \begin{split}
      \mathcal{X}_t^{k+1} &= \frac{1}{\rho^k} {\cal B}_t - \frac{\sum_{t=1}^{n_3} {\cal M}_t \odot {\cal B}_t} {(\rho^k)^2 + \rho^k  \sum_{t=1}^{n_3} {\cal M}_t^2} \odot {\cal M}_t ,
  \end{split}
    \label{eq:close_form}
\end{equation}
where ${\cal B}_t = \rho^k \left( \mathcal{V}_t^k - \mathcal{U}_t^k \right) + {\cal M}_t \odot \mathbf{Y}$ and the division between matrices is performed element-wisely. \par 
The auxiliary tensor ${\cal V}\in{\mathbb R}^{n_1\times n_2\times n_3}$ is parameterized by the GridTD model under tensor parallelization, and hence the $\mathcal{V}$-subproblem is represented as the update of GridTD model parameters. Specifically, $\mathcal{V}$ is represented by
\begin{equation}\small\label{eq:V}
{\cal V}:={\cal A}({\cal L}),\;{\cal L}[i,j,t]:=(g_\theta\circ{\bf H}_{\rm GridTD})\left(\frac{i-1}{n_1},\frac{j-1}{n_2},\frac{t-1}{n_3}\right),
\end{equation}
where ${\cal A}({\cal L})$ denotes the affine transformation \eqref{affine} derived from the latent low-rank tensor $\cal L$, and $\cal L$ is parameterized by the GridTD model $(g_\theta\circ{\bf H}_{\rm GridTD})(\cdot):{\mathbb R}^3\to{\mathbb R}$, in which $g_\theta:{\mathbb R}^R\to{\mathbb R}$ is a lightweight MLP and ${\bf H}_{\rm GridTD}:{\mathbb R}^3\to{\mathbb R}^R$ is the GridTD encoding defined in Definition \ref{def:gtd}. For the tensor data ${\cal L}$, the tensor parallelism in Lemma \ref{def:gtdp} can be employed for efficient encoding computation. Hence, the learnable parameters here include the MLP parameters in $g_\theta(\cdot)$, the grid parameters in ${\bf H}_{\rm GridTD}(\cdot)$ (stored in learnable hash tables), and the learnable parameters in the temporal affine adapter ${\cal A}(\cdot)$ including the INR parameters $f_x(\cdot),f_y(\cdot)$ and the scaling/rotation parameters. We use $\Theta$ to denote all these learnable parameters. Consequently, the $\cal V$-subproblem in the ADMM \eqref{admm_eqs} is equivalent to update the network parameters $\Theta$ as: 
\begin{equation}\small
\begin{aligned}
\min_{\Theta} \  \frac{\rho^k}{2} \left\| \mathcal{V}_\Theta - \mathcal{X}^{k+1} - \mathcal{U}^k \right\|_F^2+\lambda_1{\rm TV}({\cal V}_\Theta)+\lambda_2{\rm SSTV}({\cal V}_\Theta),
\label{eq:v_subproblem_tensor}
\end{aligned}
\end{equation}
where ${\cal V}_\Theta$ denotes the tensor obtained in \eqref{eq:V}, which is parameterized by learnable parameters $\Theta$. We introduce the smooth regularizations $\text{TV}({\cal V}_\Theta)$ and $\text{SSTV}({\cal V}_\Theta)$ defined in \eqref{SSTV} as additional priors imposed on the tensor ${\cal V}_\Theta$. The optimization of \eqref{eq:v_subproblem_tensor} can be conducted by standard gradient descent. Following related works \cite{miaotci,feng2025spatiotemporal}, we leverage the Adam optimizer to update all parameters in $\Theta$ within the $\cal V$-subproblem. Finally, the Lagrange multiplier $\cal U$ and the penalty parameter $\rho$ are updated according to \eqref{admm_eqs}.\par 
\begin{figure*}[t]
\scriptsize
    \centering
    \setlength{\tabcolsep}{2pt} 
    \begin{tabular}{cccccccc}
        \includegraphics[width=0.117\textwidth]{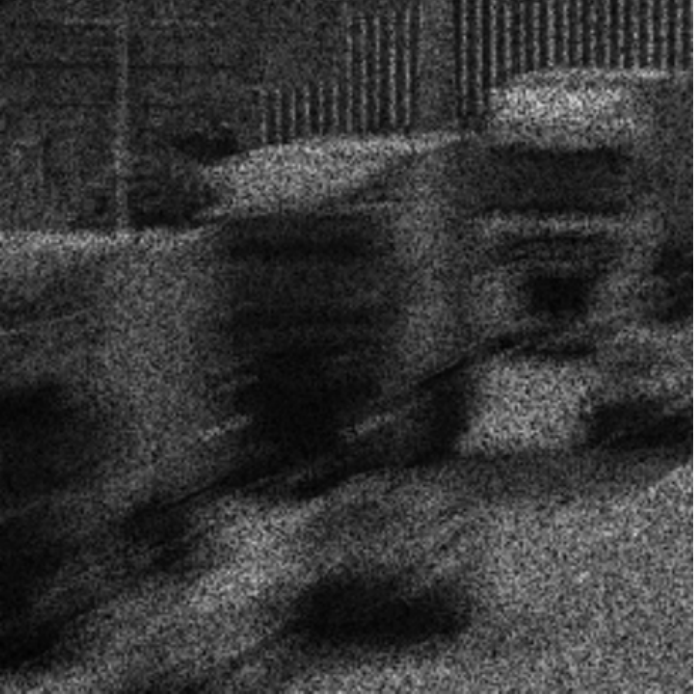} &
        \includegraphics[width=0.117\textwidth]{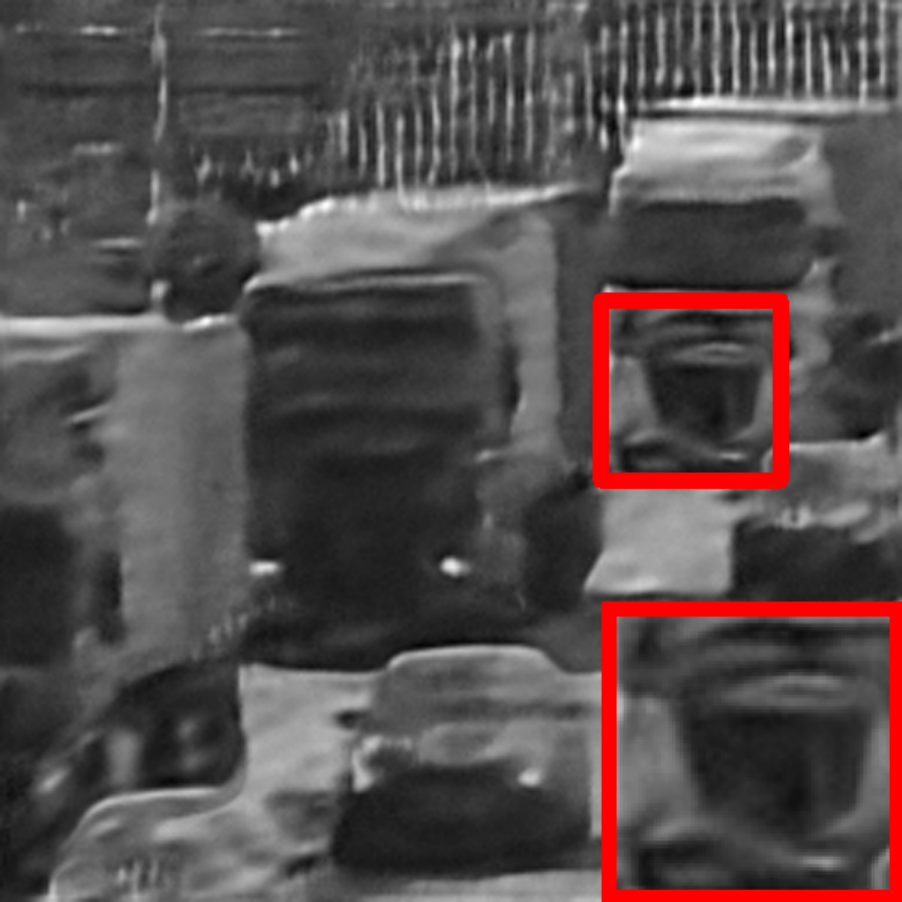} &
        \includegraphics[width=0.117\textwidth]{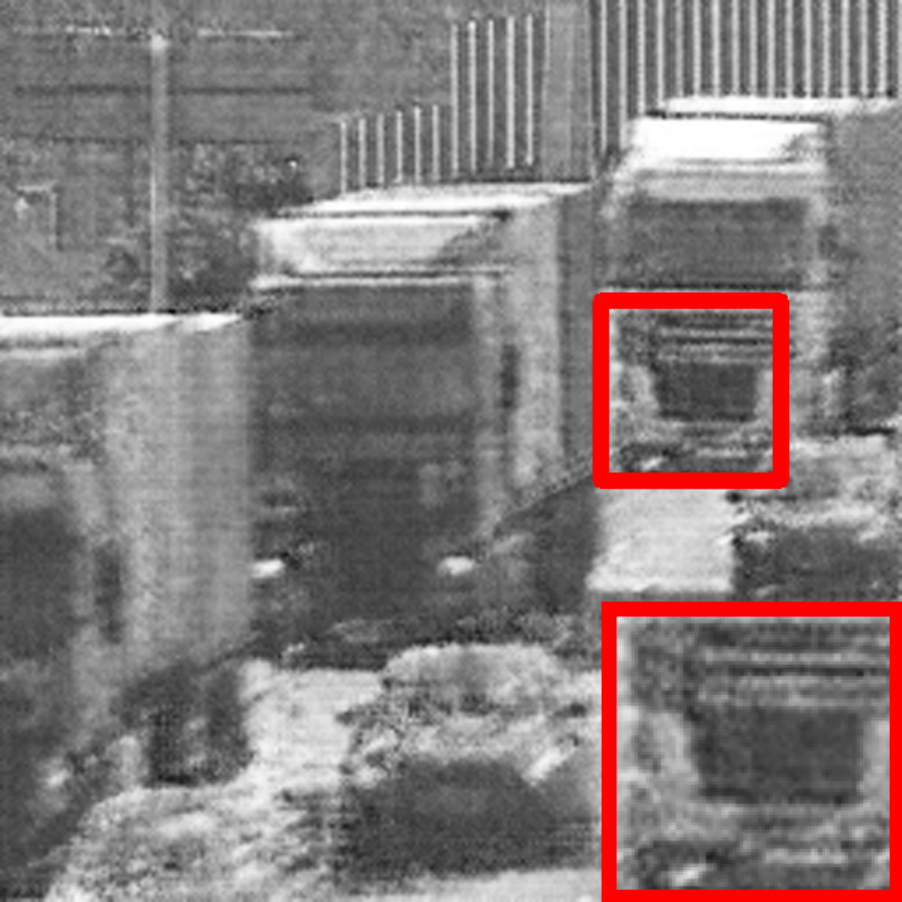} &
        \includegraphics[width=0.117\textwidth]{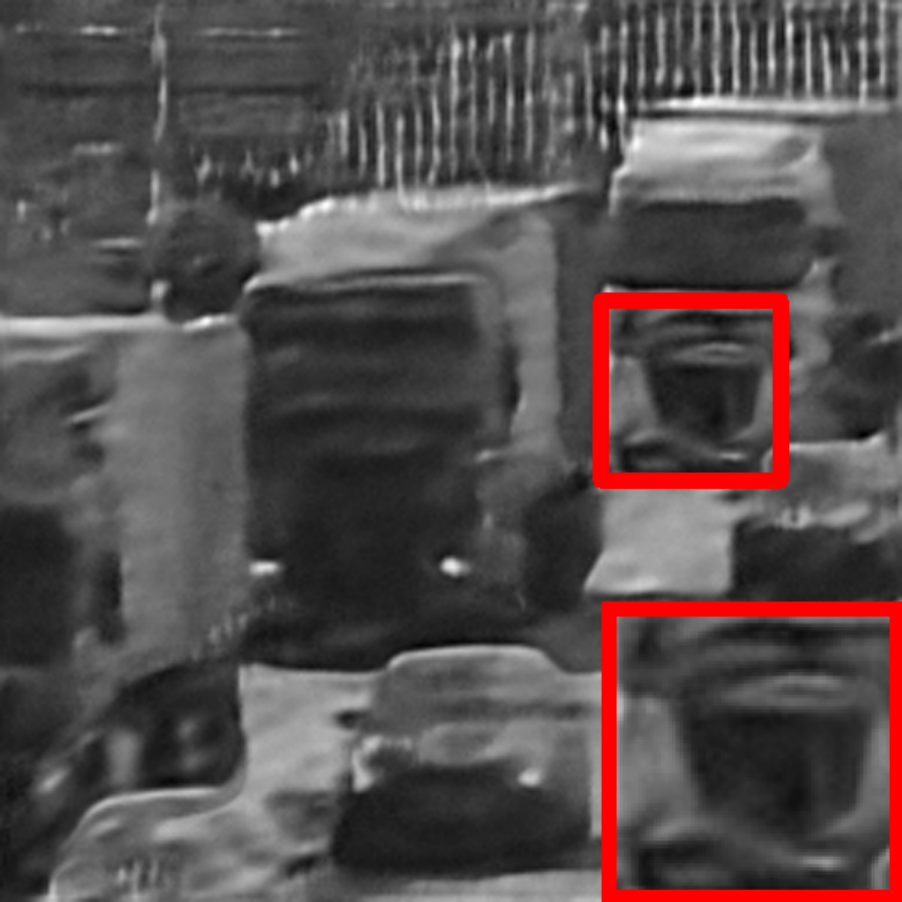} &
        \includegraphics[width=0.117\textwidth]{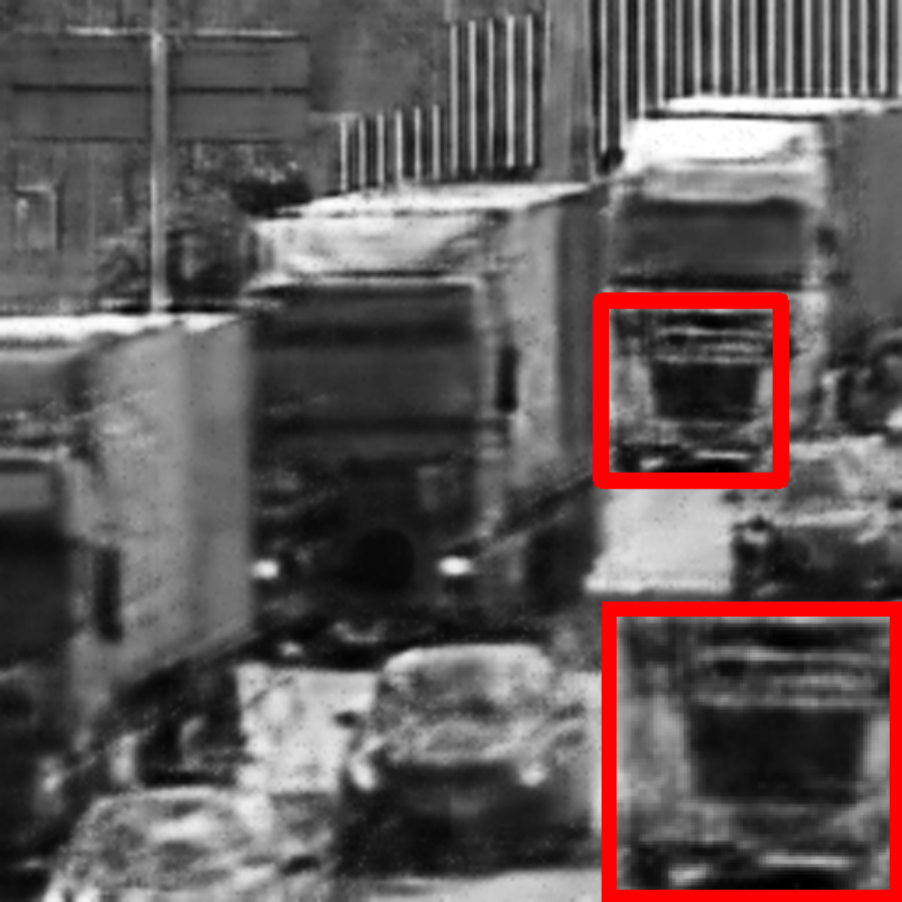}&
        \includegraphics[width=0.117\textwidth]{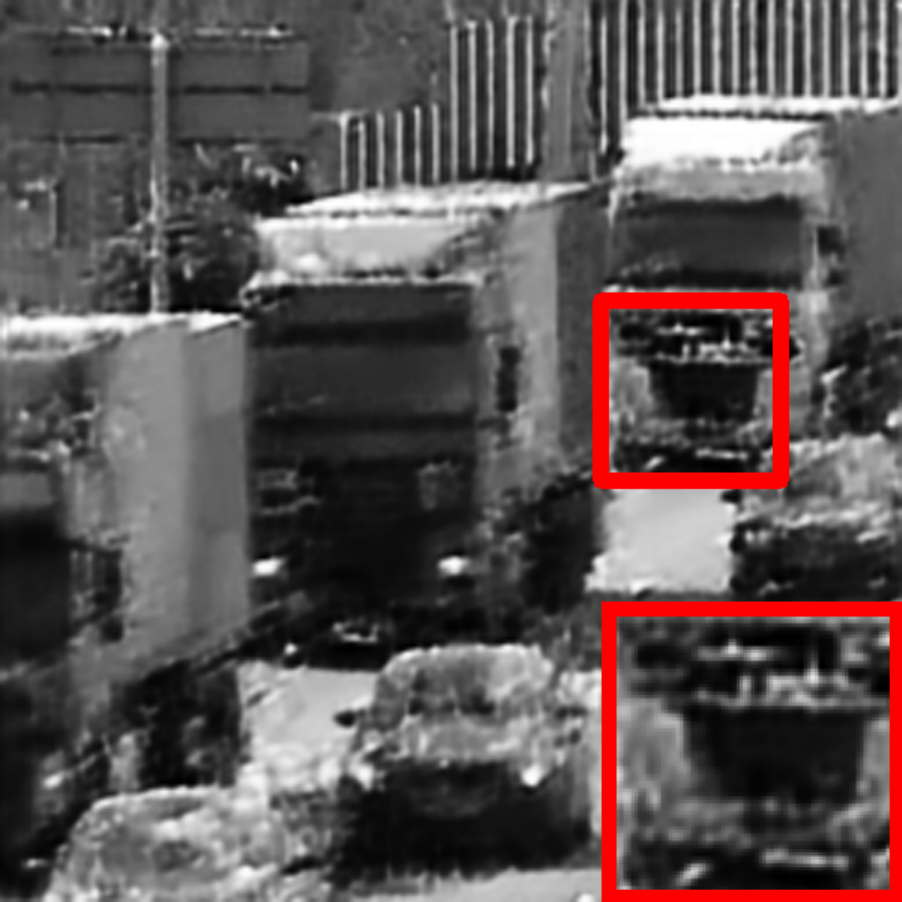} &

        \includegraphics[width=0.117\textwidth]{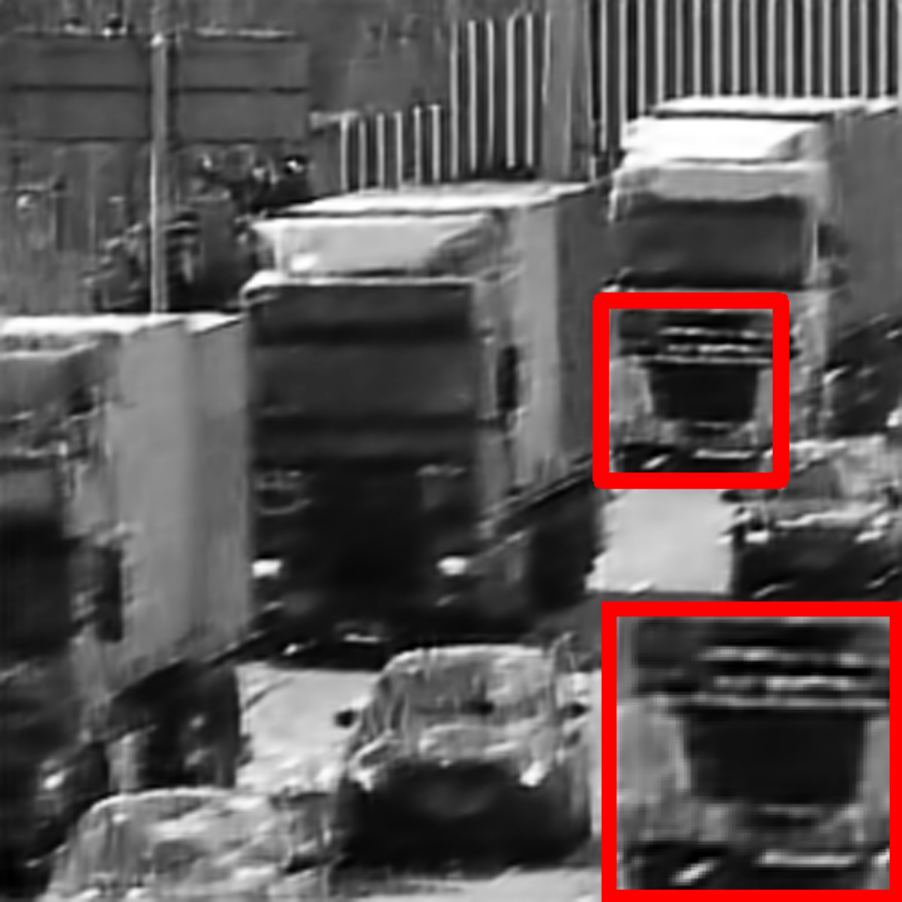} &
        \includegraphics[width=0.117\textwidth]{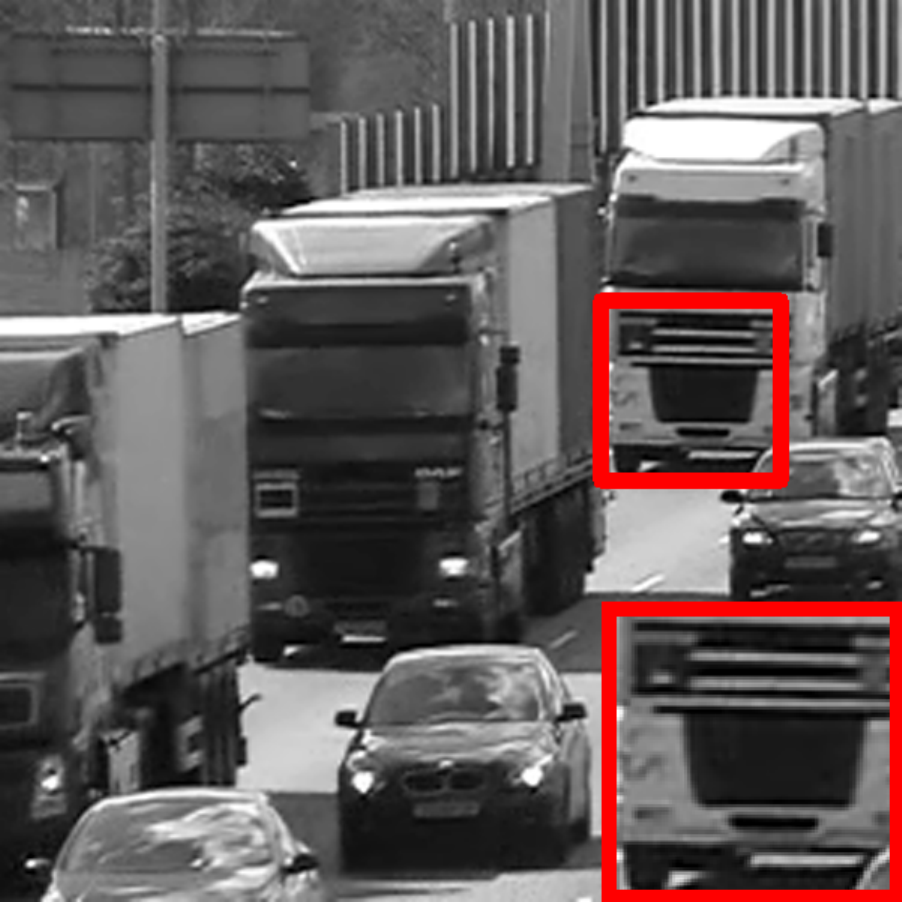} \\
        \includegraphics[width=0.117\textwidth]{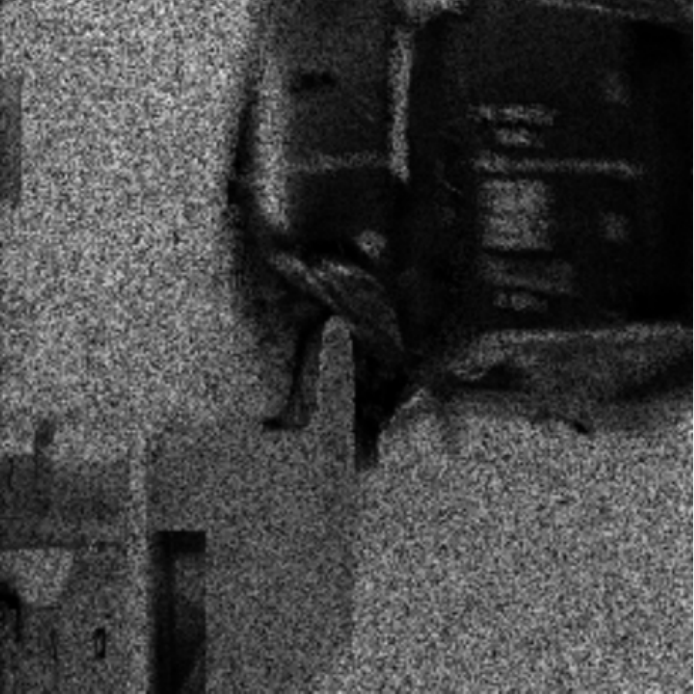} &
        \includegraphics[width=0.117\textwidth]{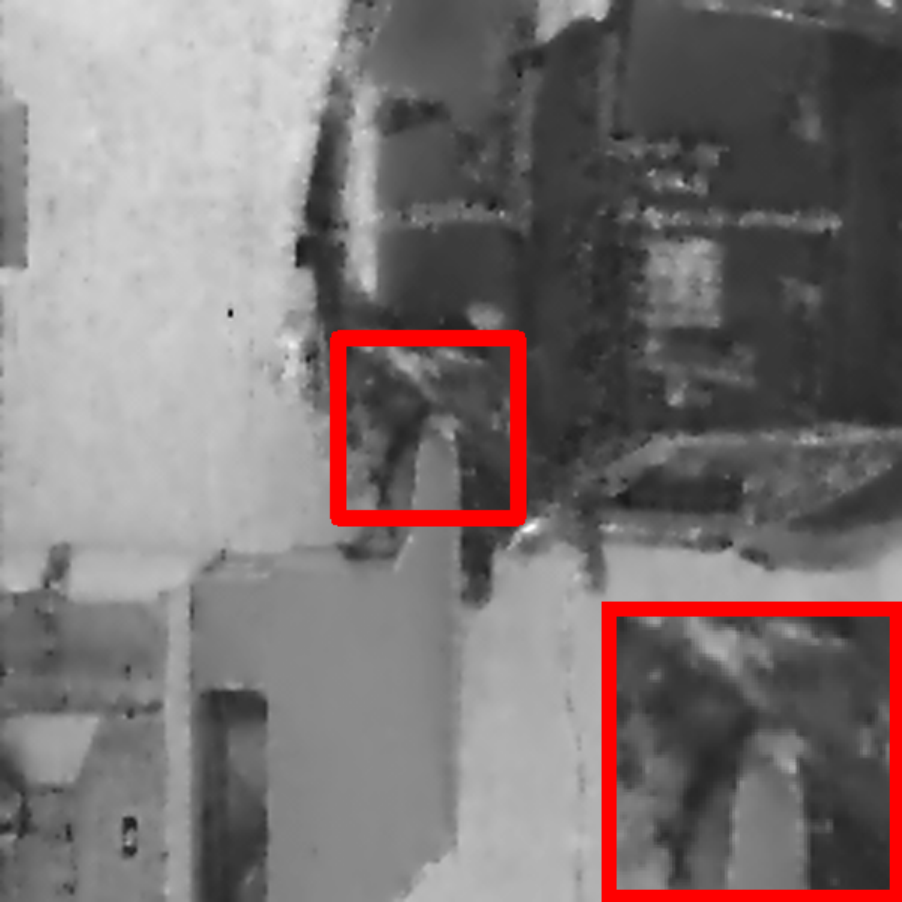} &
        \includegraphics[width=0.117\textwidth]{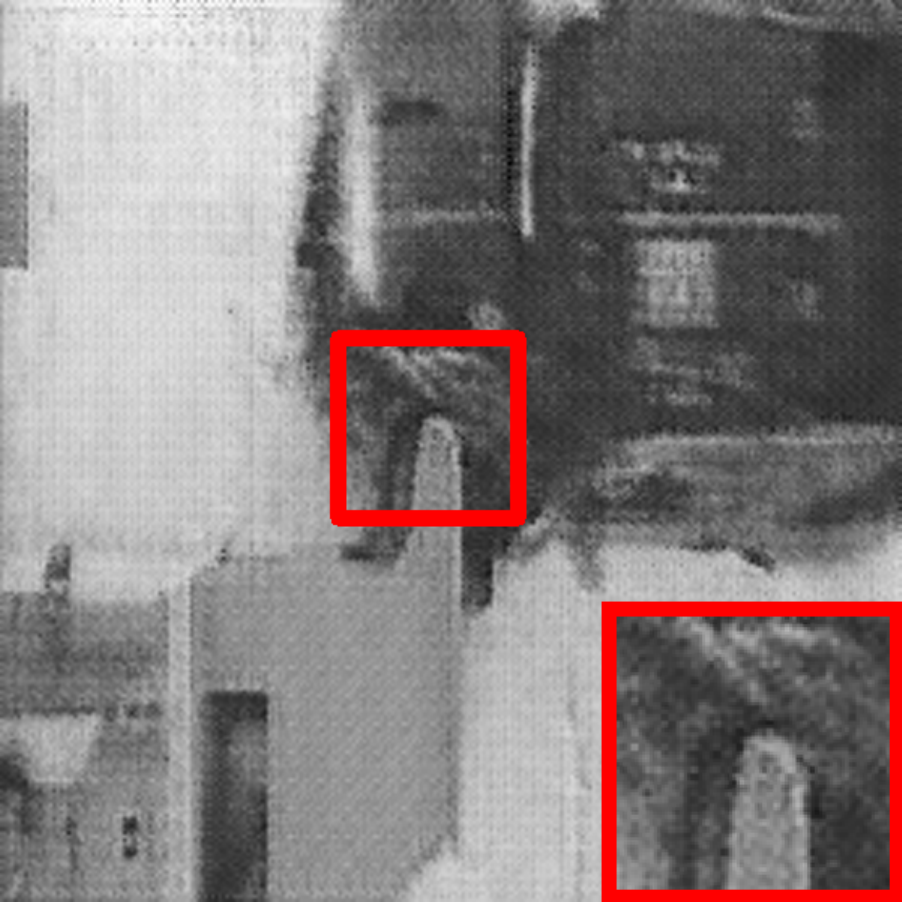} &
        \includegraphics[width=0.117\textwidth]{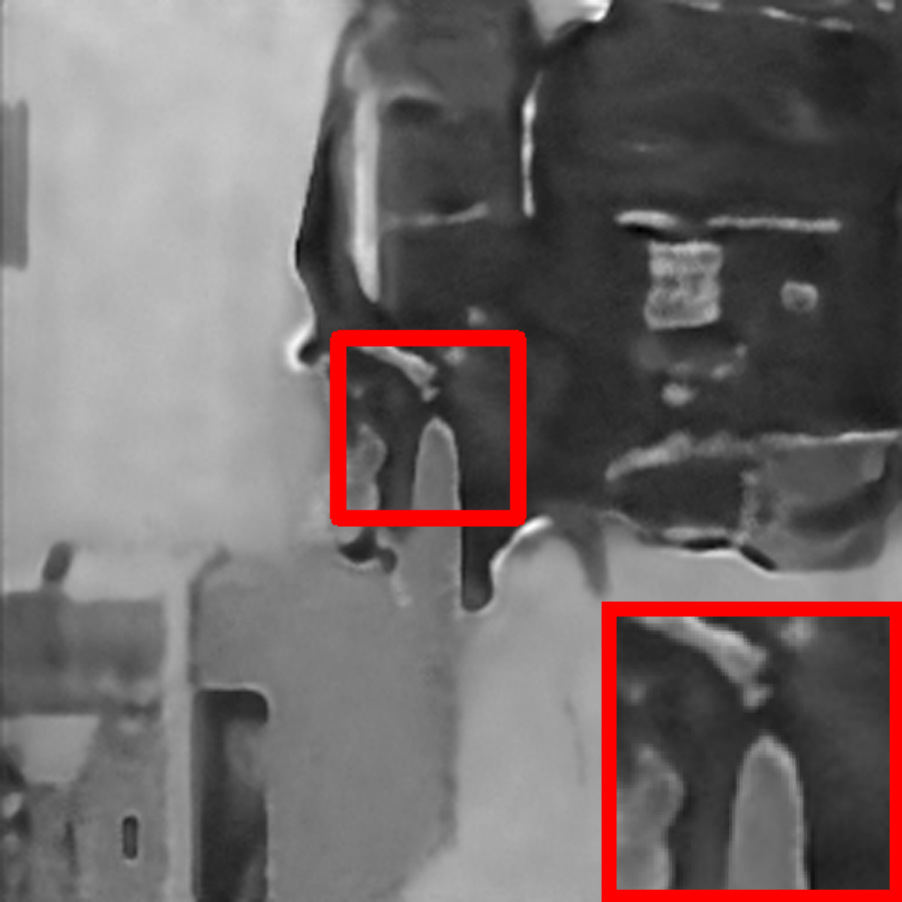} &
        \includegraphics[width=0.117\textwidth]{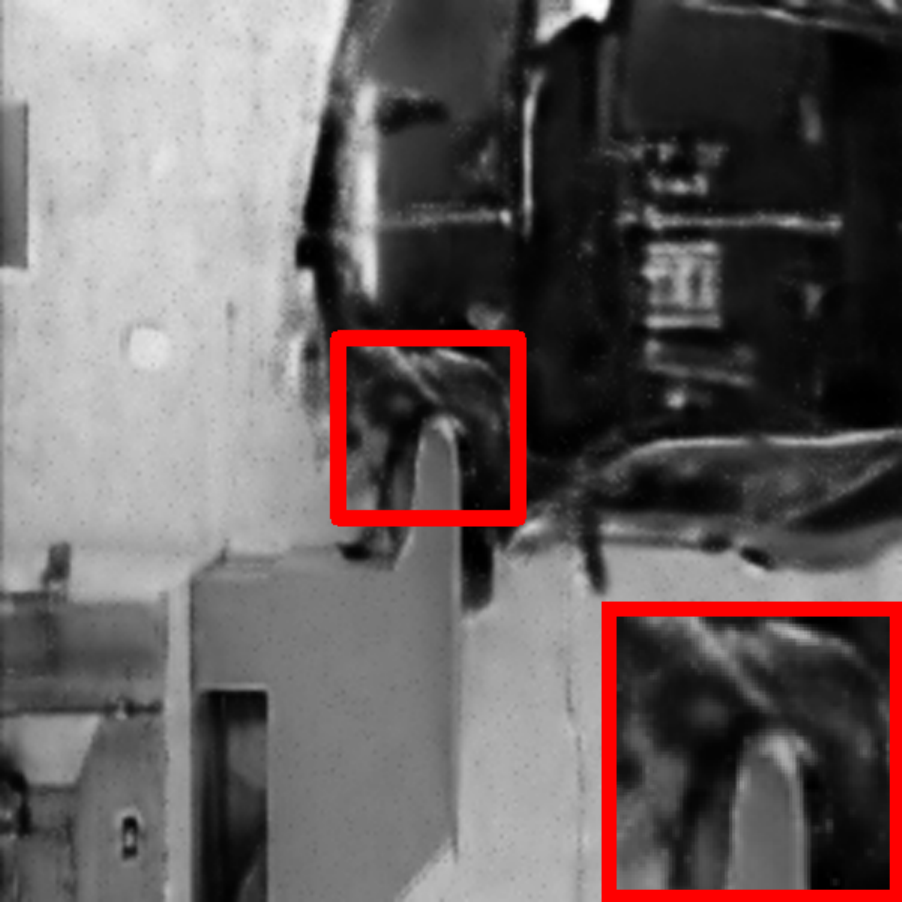}&
        \includegraphics[width=0.117\textwidth]{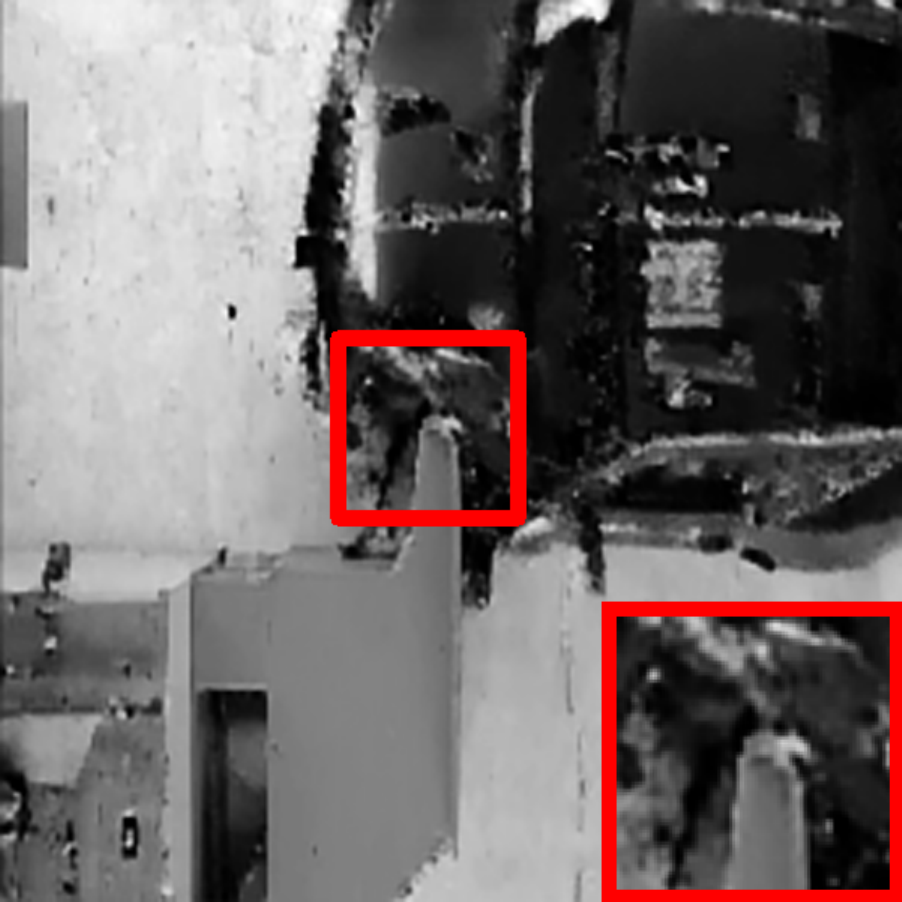} &

\includegraphics[width=0.117\textwidth]{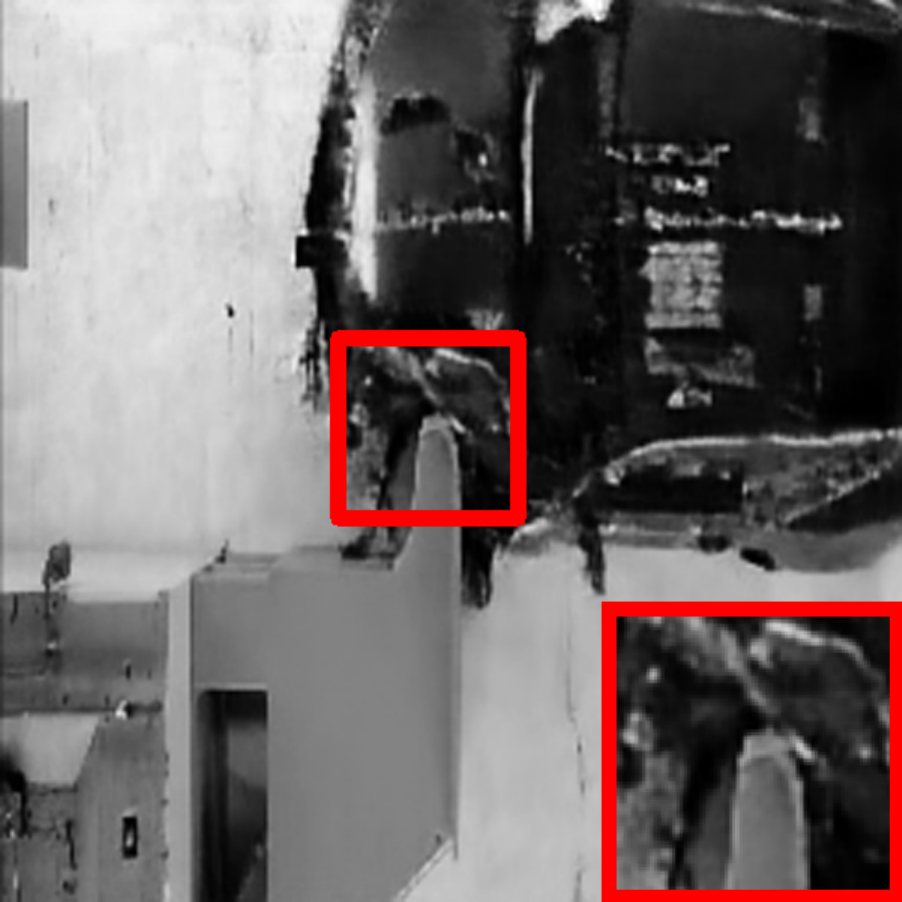} &
        \includegraphics[width=0.117\textwidth]{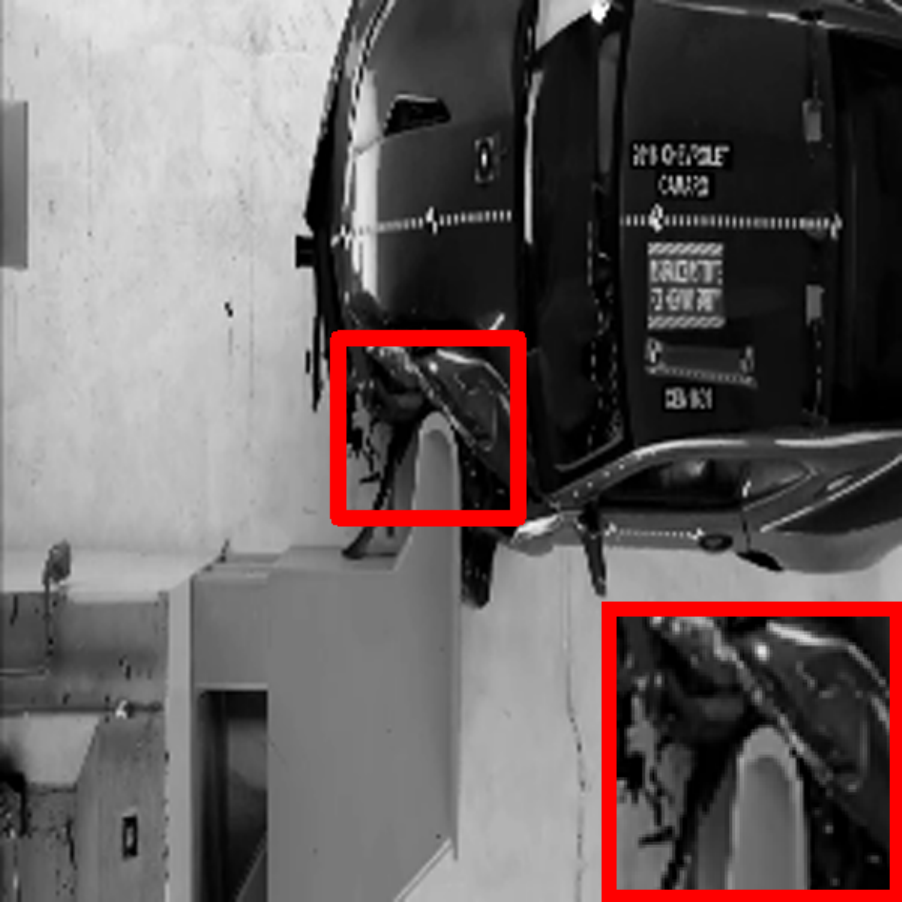} \\
\vspace{-0.2cm}
    Measurement & GAP-TV & PnP-DIP & DVP& SCI-BDVP(GAP) & InstantNGP&GridTD & Original
    \end{tabular}
    \caption{
Reconstruction comparison using various methods on two video frames for video SCI. The top row shows the results for the Traffic scene, and the bottom row shows the results for the Crash scene.
}\vspace{-0.1cm}
    \label{fig:video}
\end{figure*}
\begin{table*}[t]
\tabcolsep=7pt
\centering
\caption{Quantitative results on the six grayscale simulation benchmark videos for video SCI. (PSNR/SSIM)}
\label{table:video}
\begin{tabular}{lcccccccc}
\hline
\textbf{Method} & \textbf{Aerial} & \textbf{Crash} & \textbf{Drop} & \textbf{Kobe} & \textbf{Runner} & \textbf{Traffic} & \textbf{Average} &Average time\\ 
\hline
GAP-TV          & 25.04/0.828  & 24.67/0.827  & 34.32/0.966  & 26.71/0.843  & 29.58/0.911 & 20.76/0.702  & 26.85/0.846  &0.875s\\  
PnP-DIP         & 24.53/0.729  & 23.56/0.656  & 31.20/0.909  & 23.99/0.664  & 27.57/0.791  & 21.94/0.662  & 25.47/0.735  &429.80s \\ 
DVP  &23.04/0.689 &23.33/0.718 &31.75/0.921 &26.10/0.765 &29.09/0.842 &20.38/0.589 & 25.62/0.754  &326.68s\\
Factorized-DVP  & \underline{26.84}/\underline{0.860}  & \underline{26.05}/0.850 & 36.69/0.970  & 25.54/0.740  & 30.76/0.890  & \underline{23.38}/0.760  & 28.21/0.845  & 951.20s\\ 
SCI-BDVP (GAP)  & 26.01/0.829  & 25.57/0.819 & 40.25/\underline{0.983}  & 28.20/0.881 & 34.38/0.954  & 22.87/0.759  &  29.55/0.871 &5342.59s\\ 

InstantNGP& 26.27/0.858  & 25.56/\underline{0.873} & \underline{40.80}/\textbf{0.989}  & \underline{28.92}/\underline{0.888} & \underline{34.50}/\underline{0.955}  & 22.70/\underline{0.773}  &  \underline{29.79}/\underline{0.889} & 324.02s\\ 

{\bf GridTD}            & \textbf{27.12}/\textbf{0.873}  & \textbf{26.69}/\textbf{0.892} & \textbf{41.55}/\textbf{0.989}  & \textbf{29.23}/\textbf{0.896} & \textbf{35.22}/\textbf{0.957} &\textbf{24.00}/\textbf{0.808} & \textbf{30.64}/\textbf{0.903}  &128.72s \\ 
\hline
\end{tabular}\vspace{-0.2cm}
\end{table*}
The PnP-ADMM algorithm of GridTD for video SCI is illustrated in Algorithm \ref{alg:GridTD}. Under mild assumptions, the PnP-ADMM \eqref{admm_eqs} induced by the GridTD model admits a fixed point convergence \cite{PnP}. We first introduce the Lipschitz smoothness of the SCI fidelity term, and then illustrate the convergence result based on the Lipschitz smoothness.
\begin{lemma}[Lipschitz smoothness]\label{lemma:2}Consider the SCI fidelity term function $f(\mathcal{X}) := \frac{1}{2} \left\|\mathbf{Y} -\sum_{t=1}^{n_3} \mathcal{M}_t \odot \mathcal{X}_t \right\|_F^2$, where $\mathcal{X} \in \mathbb{R}^{n_1 \times n_2 \times n_3}$, each mask satisfies $\mathcal{M}_t \in \{0,1\}^{n_1 \times n_2}$, and $\mathbf{Y} \in \mathbb{R}^{n_1 \times n_2}$ is the observed measurement. Then the function $f(\cdot)$ is Lipschitz smooth with Lipschitz constant $
L \leq \left(\sqrt{n_3}\|\mathcal{X}\|_F + \|\mathbf{Y}\|_F\right)\sqrt{ \sum_{t=1}^{n_3} \|\mathcal{M}_t\|_F^2}.$
\end{lemma}
\begin{theorem}[Fixed point convergence of GridTD-induced PnP-ADMM \eqref{admm_eqs}]
\label{admm}
Supposed the $\mathcal{V}$-subproblem is bounded, i.e., $\left\|\mathcal{V}^{k+1} - (\mathcal{X}^{k+1} + \mathcal{U}^k) \right\|_F \leq \frac{\alpha_2}{\rho^k}$. Let the sequence generated by the ADMM algorithm in \eqref{admm_eqs} be $\left\{ \mathcal{X}^k, \mathcal{V}^k, \mathcal{U}^k \right\}$. Then, this sequence converges to a fixed point $\{\mathcal{X}^*, \mathcal{V}^*, \mathcal{U}^*\}$, i.e., as $k \to\infty$ we have $
\left\lVert \mathcal{X}^k - \mathcal{X}^* \right\rVert_F \to 0, \;
\left\lVert \mathcal{V}^k - \mathcal{V}^* \right\rVert_F \to 0, \;
\left\lVert \mathcal{U}^k - \mathcal{U}^* \right\rVert_F \to 0$.
\end{theorem}
The optimization algorithms for MRI reconstruction \eqref{eq:mri_model} and spectral SCI \eqref{eq:spectral_sci_model} are analogous to the PnP-ADMM \eqref{admm_eqs}, and the fixed-point convergence can be established likewise. The fixed-point convergence ensures the stability and robustness of the proposed GridTD-based ADMM algorithm.
\section{Experiments}
We conduct three groups of compressive imaging reconstruction tasks, including video SCI, spectral SCI, and compressive dynamic MRI reconstruction, to comprehensively verify the effectiveness of the proposed GridTD method. Due to space considerations, the {implementation details} of our method (model configurations, trade-off parameters, and optimizer settings, etc.) are provided in the supplementary file.
\subsection{Video Snapshot Compressive Imaging}
\subsubsection{Datasets and Baselines}
We use 6 grayscale simulation benchmark videos, including Aerial, Vehicle, Kobe, Drop, Runner, and Traffic~\cite{yuan2021plug} with a size of $256 \times 256 \times 8$. For the video SCI sensing mask, we sample the mask values from a Bernoulli distribution with $p=0.5$, following the implementation in recent work \cite{zhao2024untrained}. We select six representative baselines: the GAP-TV~\cite{yuan2016generalized} that combines generalized alternating projection with a TV regularization; the PnP-DIP~\cite{pnpdip} that incorporates the DIP as an image prior into the iterative reconstruction process; the DVP~\cite{dvp} that leverages DIP to capture internal structural redundancy in videos; the Factorized-DVP~\cite{miaotci} that uses a tensor decomposition DVP for video SCI; the SCI-BDVP~\cite{zhao2024untrained} that integrates bagged-DVP into the SCI problem; and the InstantNGP \cite{muller2022instant} with TV and temporal affine adapter.
\begin{figure*}[ht]
    \centering
    \scriptsize
    \setlength{\tabcolsep}{2pt} 
    \begin{tabular}{cccccc}
        \includegraphics[width=0.125\textwidth]{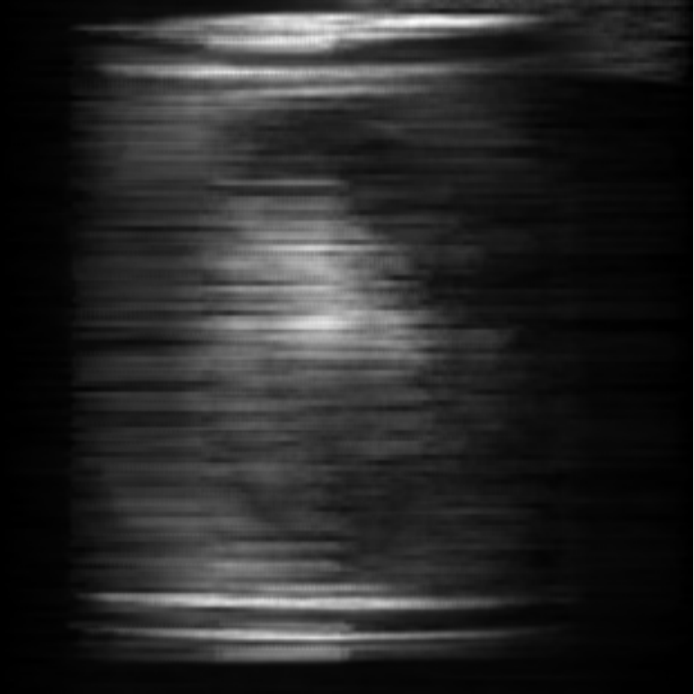} &
        \includegraphics[width=0.125\textwidth]{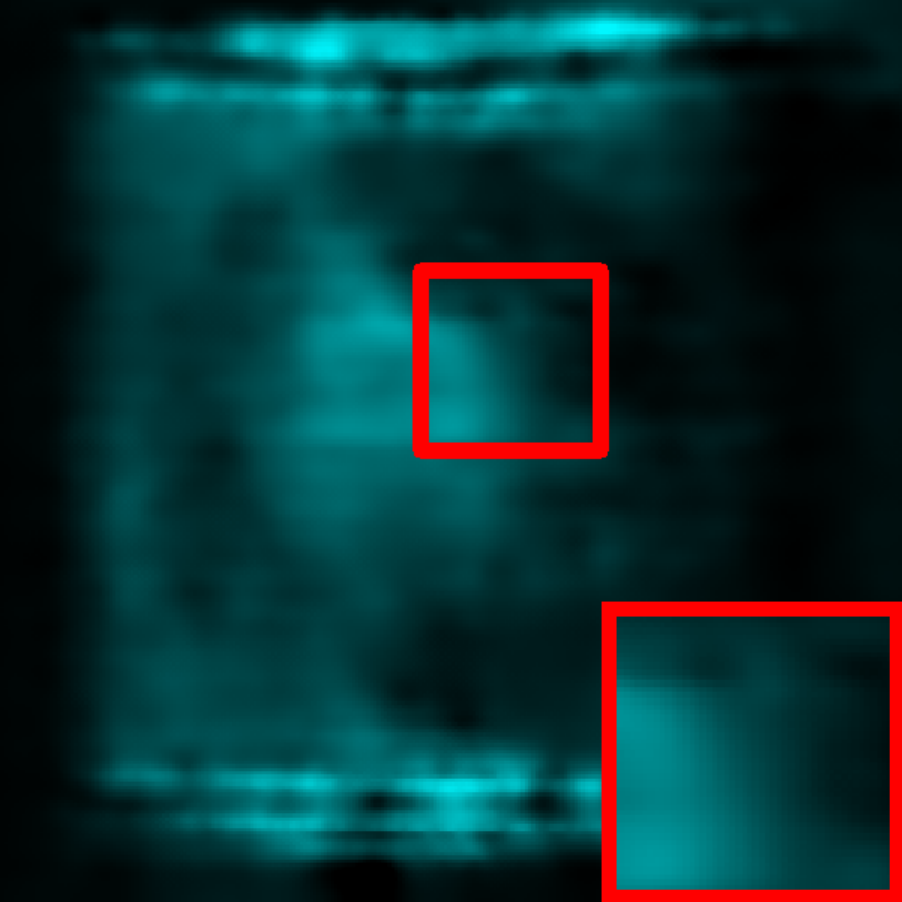} &
        \includegraphics[width=0.125\textwidth]{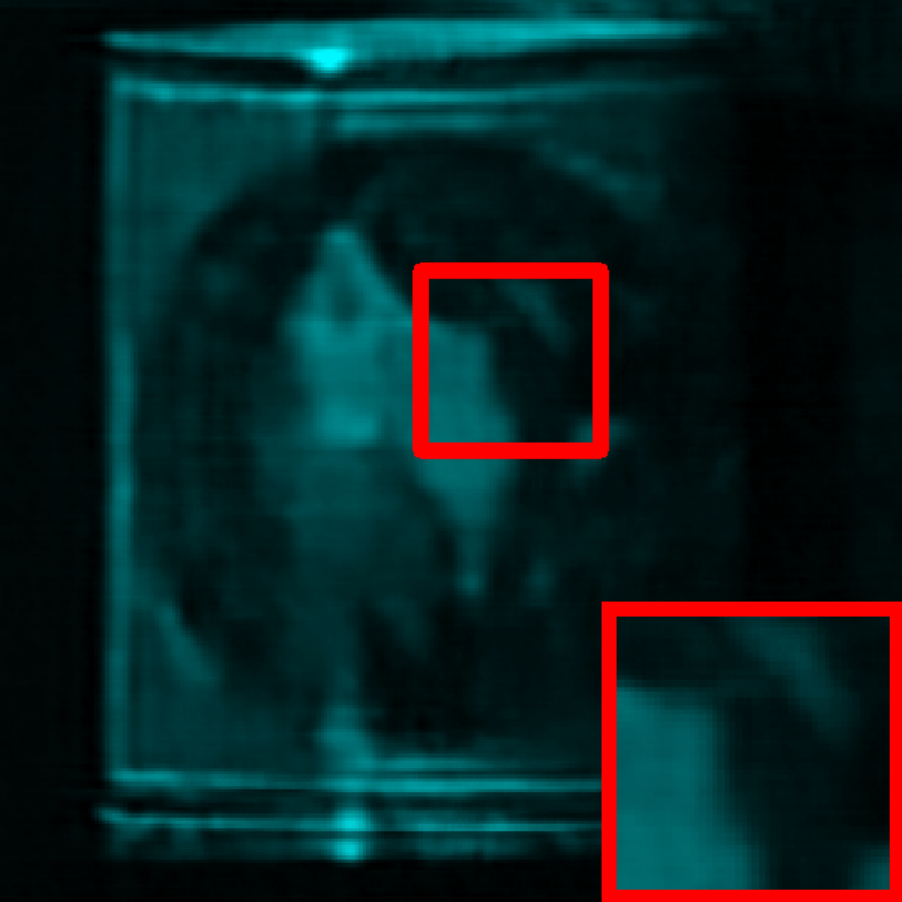} &
        \includegraphics[width=0.125\textwidth]{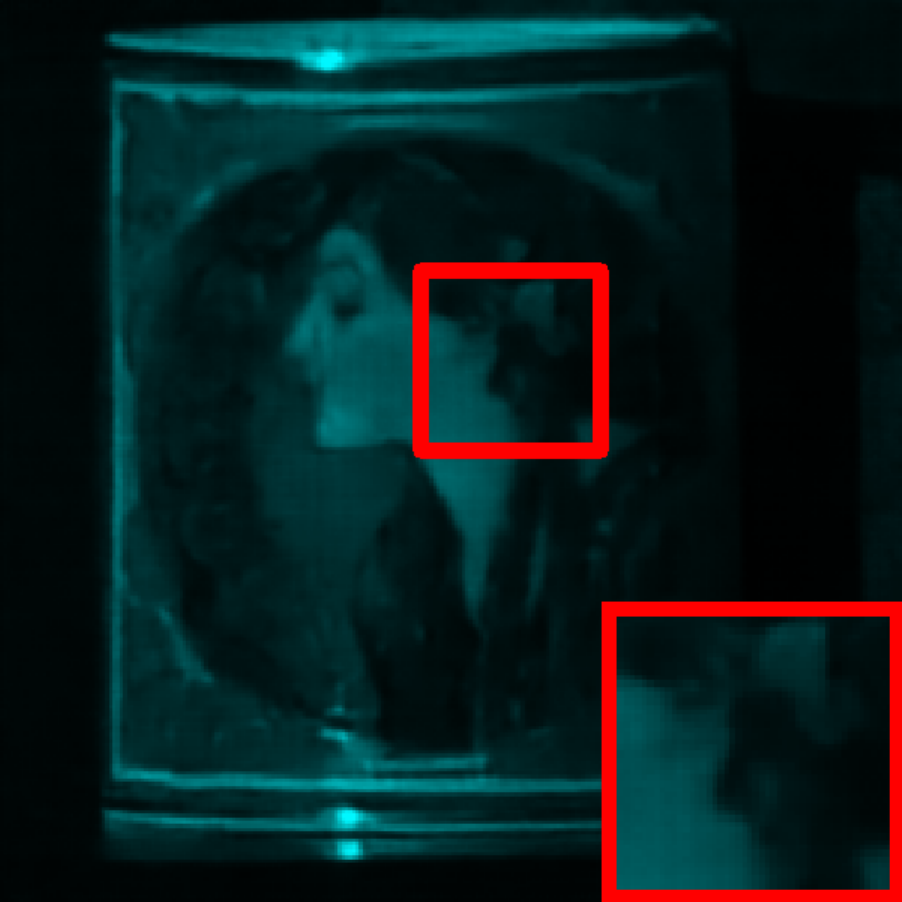} &
        \includegraphics[width=0.125\textwidth]{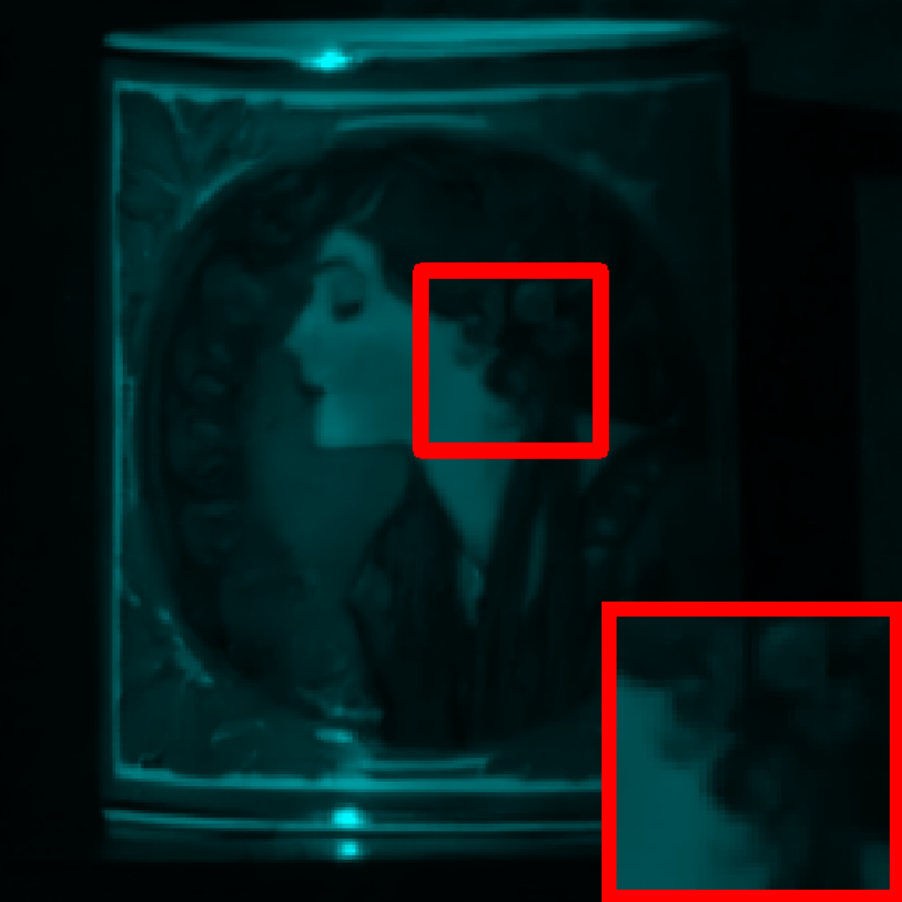} &
        \includegraphics[width=0.125\textwidth]{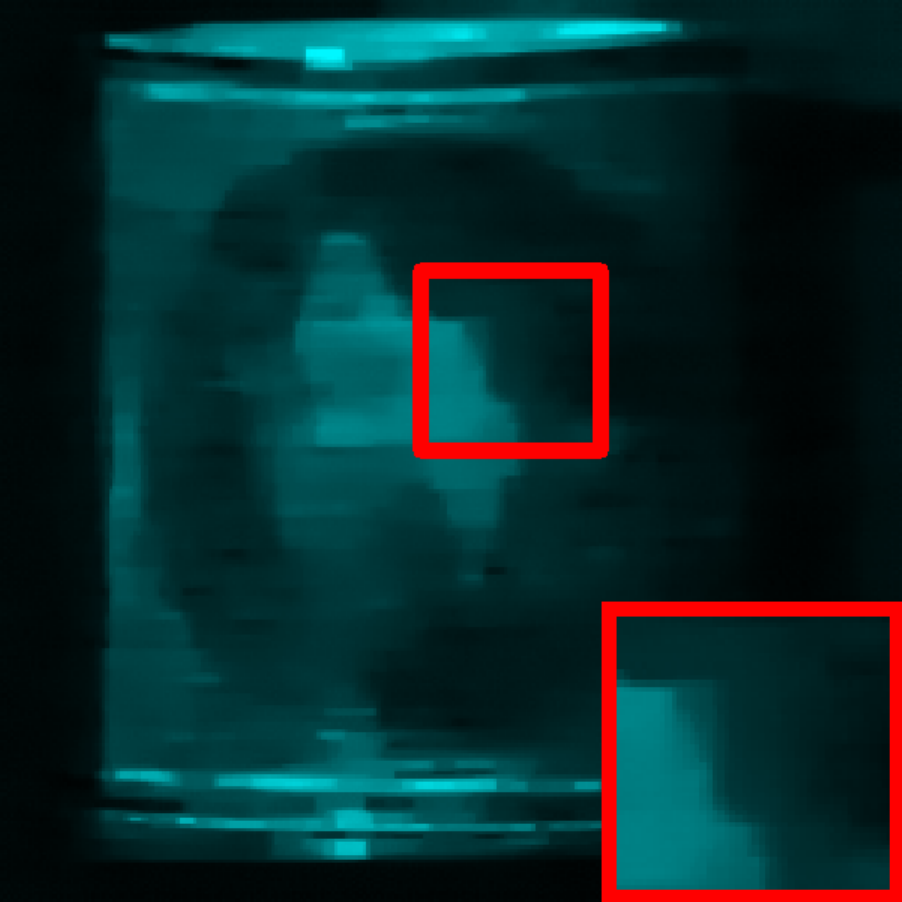}  \\
        Measurement & GAP-TV & $\lambda$-Net & TSA-Net & HDNet & HIR-Diff \\
        \includegraphics[width=0.125\textwidth]{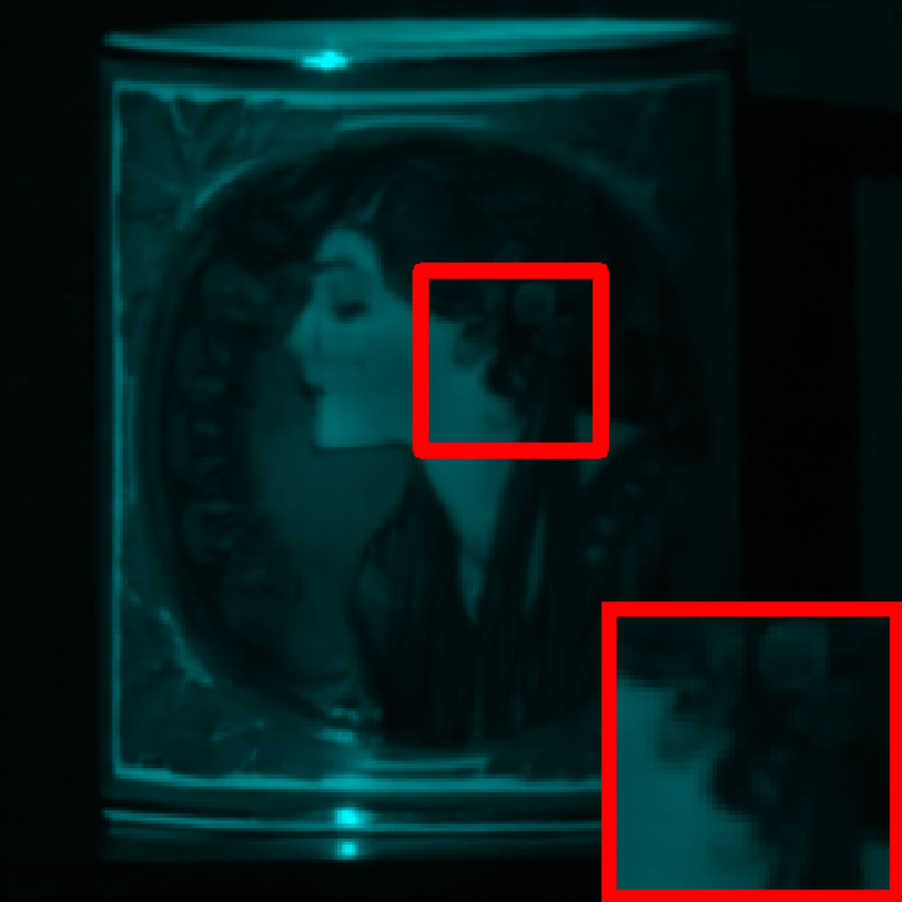} &
        \includegraphics[width=0.125\textwidth]{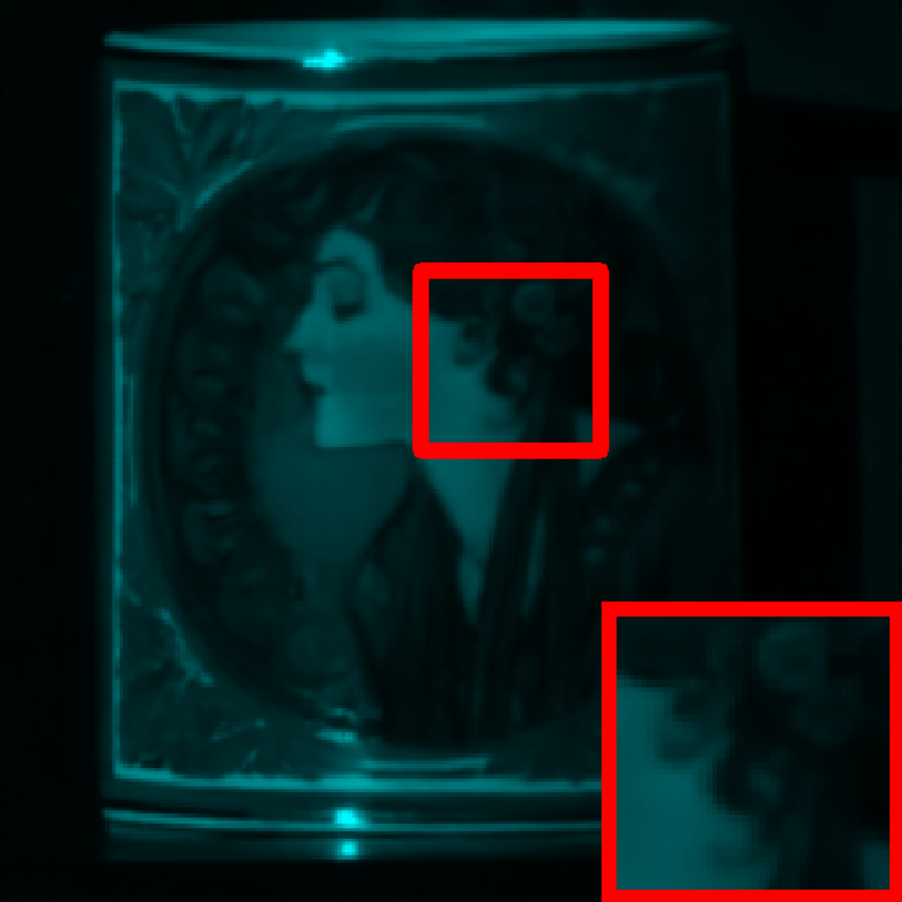} &
        \includegraphics[width=0.125\textwidth]{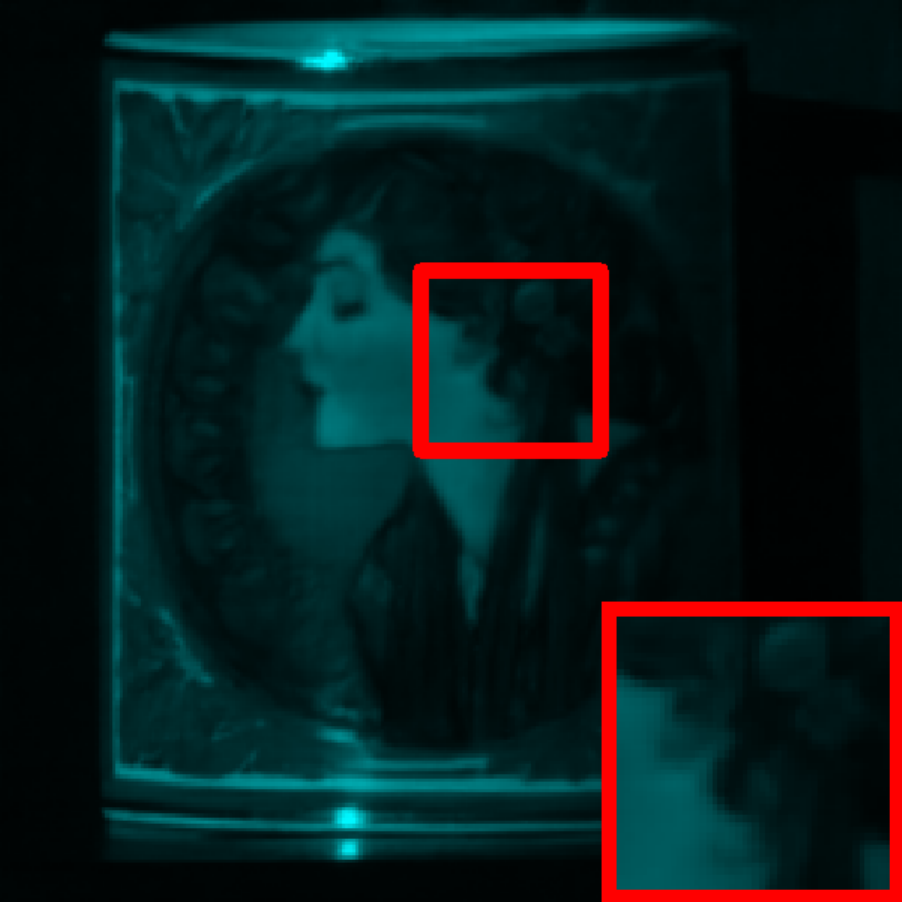} &
        \includegraphics[width=0.125\textwidth]{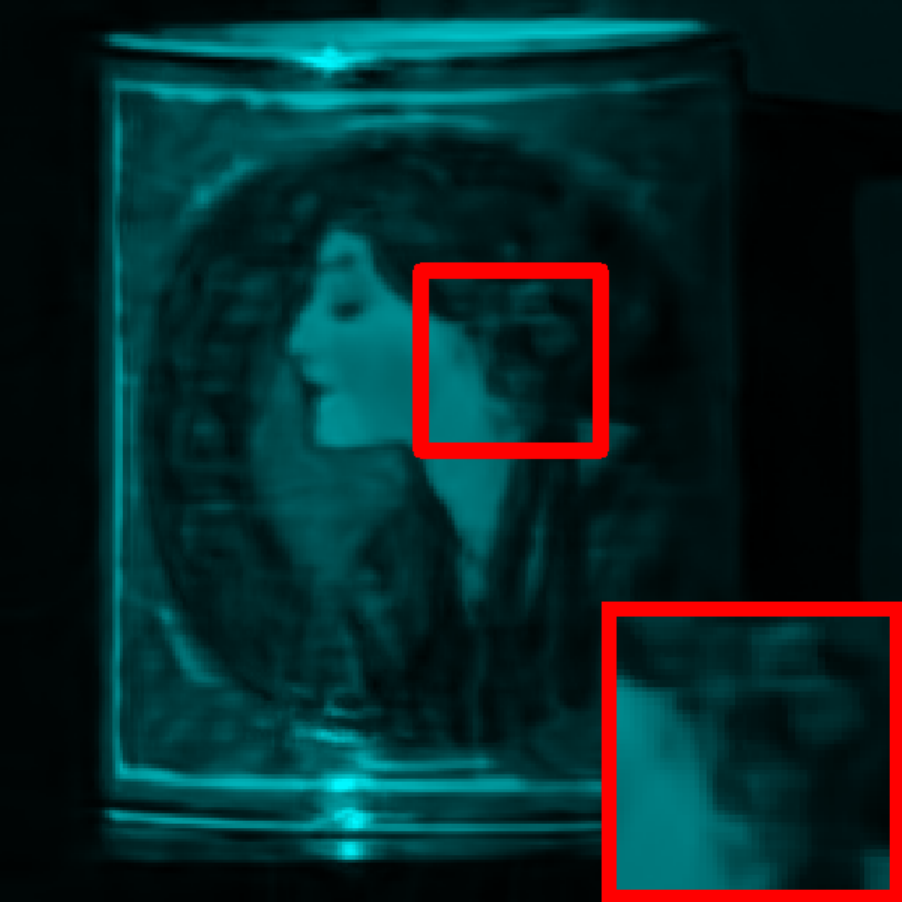} &
        \includegraphics[width=0.125\textwidth]{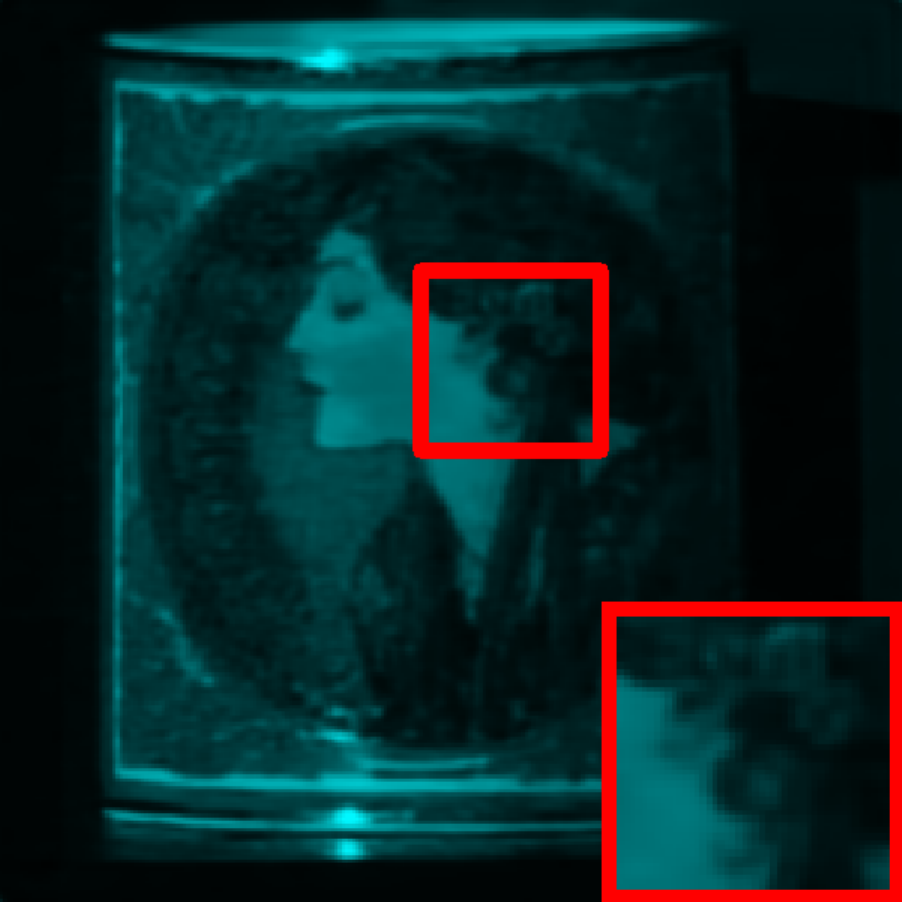} &
         \includegraphics[width=0.125\textwidth]{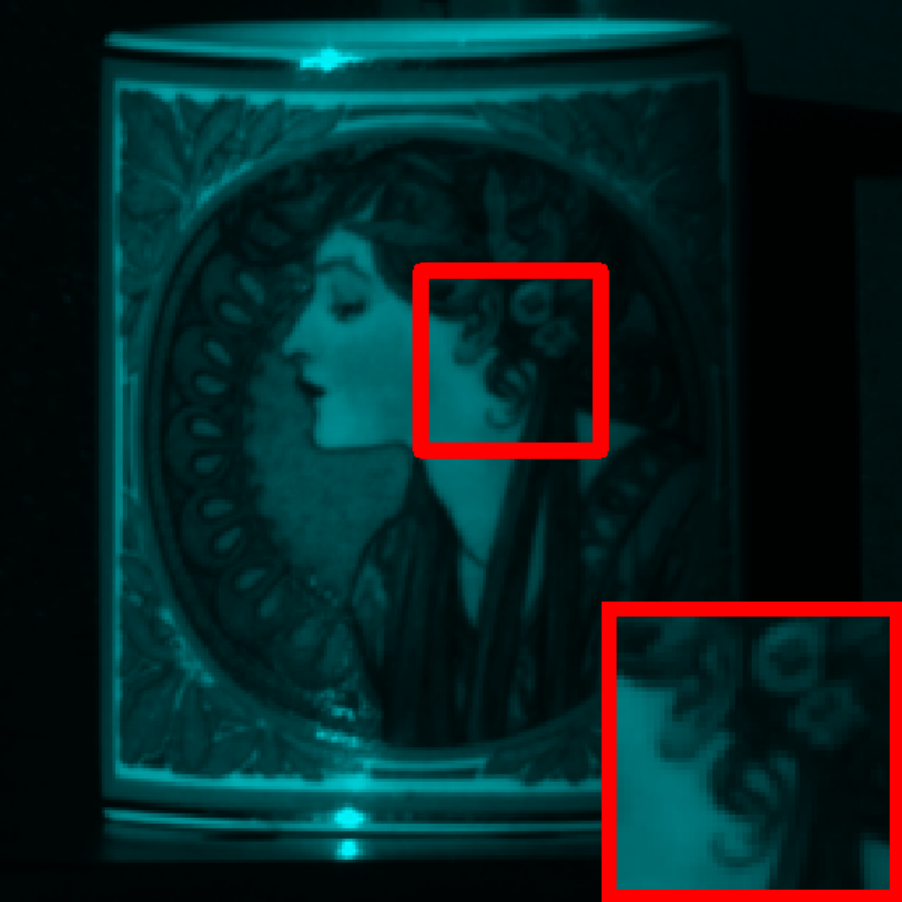} \\\vspace{-0.2cm}
         MST-L & CST-L & MST++ & LRSDN & GridTD &Original  \\
    \end{tabular}
    \caption{
Reconstruction comparison using various methods on Scene 1 from the KAIST dataset for spectral SCI.
}\vspace{-0.1cm}
    \label{fig:spectral}
\end{figure*}
\begin{table*}[t]
\centering
\caption{Quantitative results on the KAIST dataset for spectral SCI. Except for GAP-TV, HIR-Diff, LRSDN, InstantNGP, and the proposed GridTD, all other comparison methods are supervised deep learning methods. (PSNR/SSIM)}
\label{table:spectral}
\setlength{\tabcolsep}{2pt}
\begin{tabular}{lccccccccccc}
\hline
\textbf{Method}   & \textbf{Scene1}        & \textbf{Scene2}        & \textbf{Scene3}        & \textbf{Scene4}        & \textbf{Scene5}        & \textbf{Scene6}        & \textbf{Scene7}        & \textbf{Scene8}       & \textbf{Scene9}       & \textbf{Scene10}       & \textbf{Average}       \\
\hline
$\lambda$-Net & 29.75/0.818 & 27.74/0.732 & 29.58/0.834 & 37.65/0.911 & 26.64/0.777 & 27.29/0.783 & 26.65/0.766 & 26.24/0.780 & 28.58/0.785 & 26.18/0.695 & 28.63/0.788 \\
TSA-Net   & 32.37/0.891 & 31.12/0.860 & 32.35/0.911 & 39.72/0.957 & 29.48/0.879 & 31.11/0.903 & 30.30/0.875 & 29.35/0.889 & 31.67/0.902 & 29.24/0.869 & 31.67/0.894 \\
HDNet     & 35.17/0.936 & 35.73/0.942 & 36.13/0.942 & 42.78/0.976 & 32.72/0.946 & 34.52/0.955 & 33.70/0.923 & 32.49/0.947 & 34.93/0.941 & 32.39/0.944 & 35.06/0.945 \\
MST-L     & 35.51/0.943 & 36.19/0.948 & 36.48/0.951 & 42.36/0.976 & 33.00/0.947 & 34.77/0.956 & 34.15/0.926 & 32.94/0.952 & 35.12/0.940 & 32.80/0.947 & 35.33/0.949 \\
CST-L     & \underline{35.86}/\textbf{0.948} & \textbf{36.57}/\textbf{0.954} & 37.54/\underline{0.960} & 42.58/0.978 & 33.43/\underline{0.952} & \textbf{35.64}/\textbf{0.965} & 34.46/0.935 & \textbf{33.80}/\textbf{0.964} & 35.94/0.951 & \textbf{33.43}/\textbf{0.956} & 35.93/\textbf{0.956} \\
MST++     & 35.84/\underline{0.944} & \underline{36.29}/\underline{0.949} & 37.46/0.957 & 44.28/0.981 & 33.42/\underline{0.952} & \underline{35.46}/\underline{0.960} & 34.38/0.932 & \underline{33.72}/\underline{0.958} & 36.72/0.953 & \underline{33.40}/\underline{0.953} & \underline{36.10}/\underline{0.954} \\
\midrule
GAP-TV    & 25.92/0.700 & 24.73/0.610 & 25.71/0.735 & 36.19/0.866 & 23.18/0.651 & 22.41/0.617 & 23.43/0.643 & 22.56/0.607 & 24.83/0.683 & 24.29/0.525 & 25.33/0.664 \\
HIR-Diff & 29.22/0.806 & 26.79/0.724 & 30.21/0.867 & 42.81/0.962 & 26.52/0.804 & 26.76/0.675 & 27.43/0.831 & 25.70/0.756 & 27.89/0.810 & 26.32/0.692 & 28.97/0.793\\
LRSDN     & 35.44/0.922 & 34.89/0.910 & \underline{38.90}/0.959 & 45.29/0.986 &\underline{34.71}/0.948 & 33.18/0.932 & \underline{37.76}/\underline{0.962} & 30.57/0.905 & \underline{39.49}/0.962 & 30.62/0.896 & 36.09/0.938 \\
InstantNGP & 34.22/0.924 & 35.48/0.936 & 36.71/0.950 & \underline{46.09}/\underline{0.987} &34.45/0.954 & 33.01/0.945 & 36.77/0.959 & 29.74/0.918 & 39.11/\underline{0.965} & 31.06/0.911 & 35.66/0.945 \\
{\bf GridTD} & \textbf{35.98}/0.937 & 35.99/0.937 & \textbf{39.04}/\textbf{0.964} & \textbf{47.05}/\textbf{0.989} & \textbf{35.26}/\textbf{0.958} & 34.41/0.955 & \textbf{37.97}/\textbf{0.965} & 31.56/0.928 & \textbf{40.73}/\textbf{0.974} & 31.89/0.927 & \textbf{36.99}/0.953 \\
\hline
\end{tabular}
\end{table*}
\begin{figure*}[h]
    \centering
    \begin{tabular}{cc}
    \includegraphics[width=0.43\textwidth,height=0.24\textwidth]{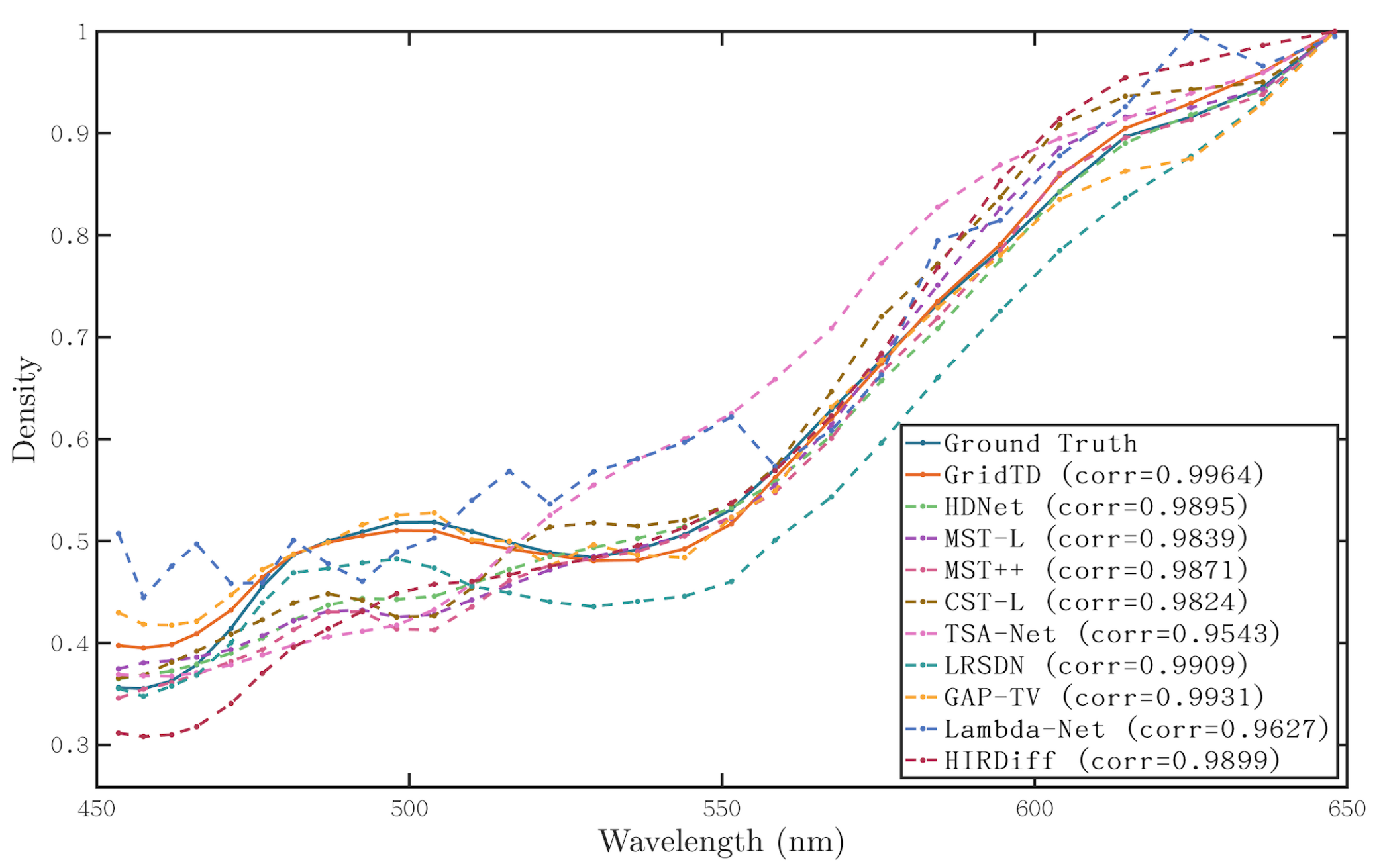} &
        \includegraphics[width=0.43\textwidth,height=0.24\textwidth]{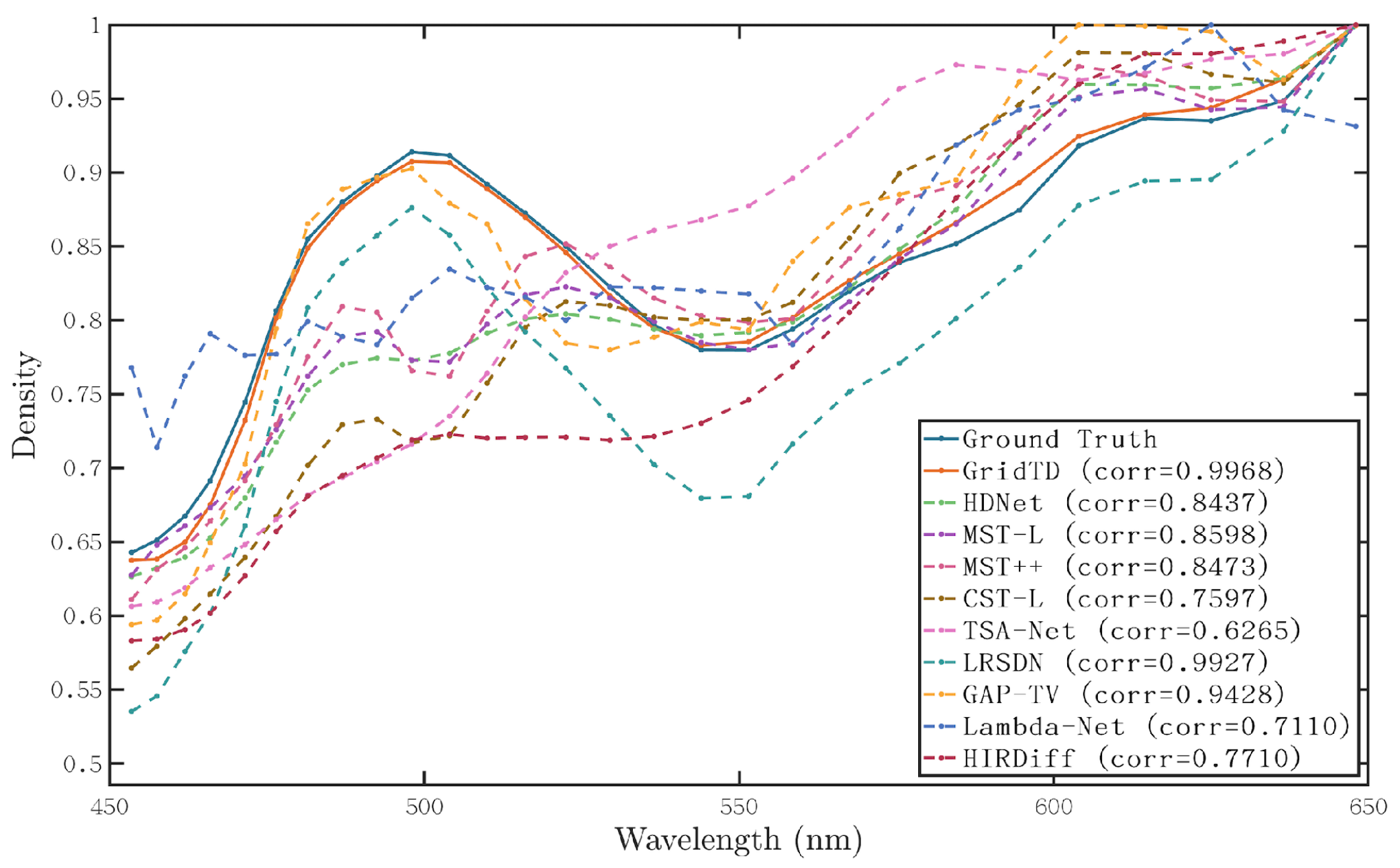} \\
    \end{tabular}\vspace{-0.2cm}
    \caption{The recovered spectral curves and the corresponding correlation coefficients of different methods for spectral SCI reconstruction on two local patches of {\bf Scene 1}.\vspace{-0.3cm}}
    \label{fig:spectral_curve}
\end{figure*}
\subsubsection{Results}
As shown in Table~\ref{table:video}, our method achieves superior results in both peak signal-to-noise ratio (PSNR) and structure similarity (SSIM) metrics compared to all other unsupervised video SCI methods. On these standard benchmark datasets, our GridTD improves the average PSNR by approximately 1 dB over the best existing unsupervised methods, with consistent improvements observed in SSIM as well. Moreover, from Table \ref{table:video}, we see that GridTD is much faster than other DIP-based SCI methods due to the compact structure of GridTD. As illustrated in Fig. \ref{fig:video}, GridTD produces clearer edge details, particularly around image boundaries and textured regions. This is attributed to the multi-resolution grid encoding that enhances the representation ability. In addition, GridTD exhibits fewer visually perceptible motion blur artifacts. This performance gain is due to the temporal affine adapter, which effectively captures motion information, thereby enhancing robustness for dynamic scenes.
\subsection{Spectral Snapshot Compressive Imaging}
\subsubsection{Datasets and Baselines}
For spectral SCI, we select 10 representative test scenes from the KAIST hyperspectral dataset~\cite{kaist}, each with high spatial and spectral resolution. The coding mask settings we adopted are the same as those used in \cite{mst}. We retain these settings in our study to ensure reproducibility of results and fairness in comparisons. We compare our method with ten mainstream reconstruction algorithms, which represent a range of paradigms from traditional optimization to end-to-end deep learning, model-driven deep unfolding, and pre-trained models. These include {\bf supervised methods} $\lambda$-Net~\cite{lambda} (attention mechanism-based network), HDNet~\cite{hdnet}, TSA-Net~\cite{tsa}, MST-L~\cite{mst}, CST-L~\cite{cst}, and MST++~\cite{mst++} (transformer-based architectures), and {\bf unsupervised methods} GAP-TV~\cite{yuan2016generalized},  LRSDN~\cite{lrsdn} (integrates deep priors with low-rank modeling), HIR-Diff~\cite{pang2024hir} (pre-trained diffusion model-based method), and InstantNGP \cite{muller2022instant} with TV. For fair comparison, all baselines are evaluated using their pre-trained checkpoints from the MST toolbox\footnote{\url{https://github.com/caiyuanhao1998/MST}}.
\subsubsection{Results}
As shown in Table~\ref{table:spectral}, the proposed method demonstrates competitive performance on the KAIST test dataset. Notably, compared to the state-of-the-art unsupervised method LRSDN, our method achieves an average PSNR improvement of approximately 1 dB. Moreover, our method compares favorably against powerful supervised approaches (including powerful transformer-based methods MST++ and CST-L), further validating the effectiveness and potential of the proposed unsupervised framework. As illustrated in Fig. \ref{fig:spectral}, our method excels in recovering high-frequency details of images, outperforming other approaches. The comparison results reveal that existing methods often produce overly smoothed image structures. In contrast, our method shows a stronger ability to accurately preserve edge details and texture contours. This performance gain is primarily attributed to the powerful multi-resolution grid encoding. From Fig. \ref{fig:spectral_curve}, we can see that GridTD accurately recovers spectral curves of the high-dimensional HSIs as compared with other methods, demonstrating the spectral reconstruction ability of GridTD. 
\subsection{Compressive Sensing Dynamic MRI Reconstruction}
\subsubsection{Datasets and Baselines}
The proposed method is tested on the retrospective cardiac cine dataset OCMR\cite{chen2020ocmr}. The detailed information of the dataset is listed in the supplementary file. The acquired data are transformed into the $k$-space domain using the multi-channel NUFFT implementation with golden-angle radial trajectories based on Fibonacci number sequencing. Coil sensitivity maps are derived through the ESPRIT algorithm\cite{espirit} to enable parallel imaging reconstruction. Different acceleration factors (AFs) are simulated for evaluation, including 21, 13, 8, 5, 3, 2, 1 spokes per frame (AF=9.7, 15.7, 25.5, 40.8, 67.9, 101.9, 203.7).\par
\begin{figure*}[t]
\scriptsize
    \centering
    \setlength{\tabcolsep}{2pt} 
    \begin{tabular}{cccccc}
        \includegraphics[width=0.125\textwidth]{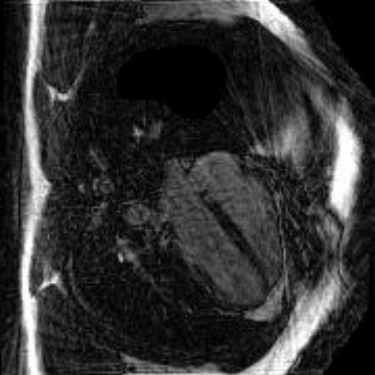} &
        \includegraphics[width=0.125\textwidth]{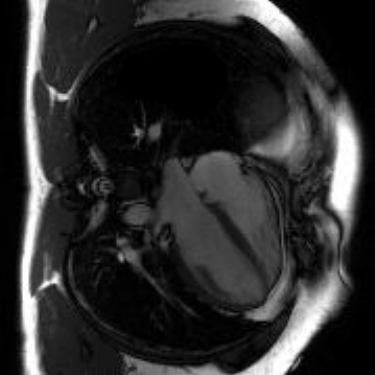} &
        \includegraphics[width=0.125\textwidth]{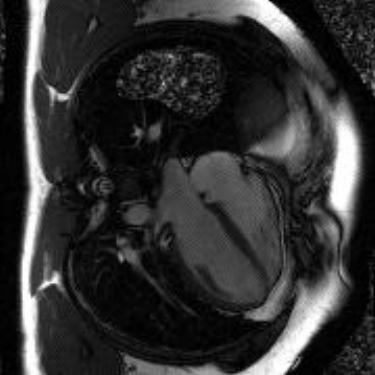} &
        \includegraphics[width=0.125\textwidth]{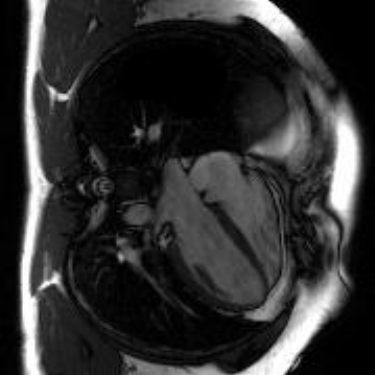}&
        \includegraphics[width=0.125\textwidth]{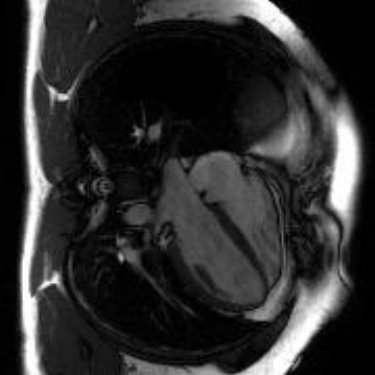} &
        \includegraphics[width=0.125\textwidth]{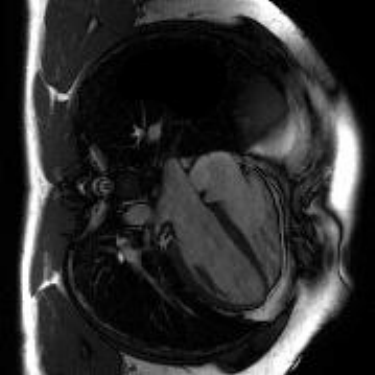} \\
        \includegraphics[width=0.125\textwidth]{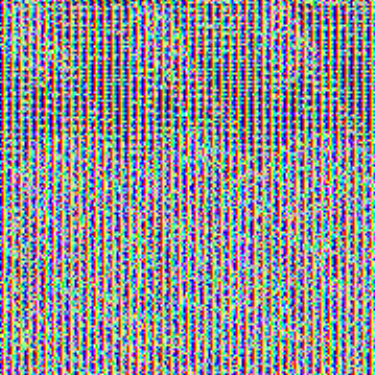} &
        \includegraphics[width=0.125\textwidth]{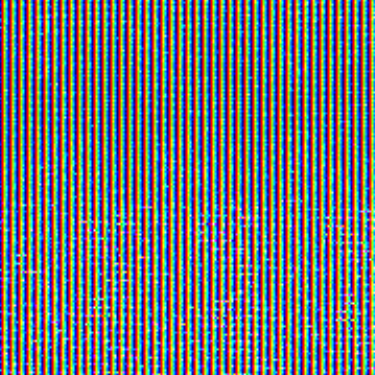} &
        \includegraphics[width=0.125\textwidth]{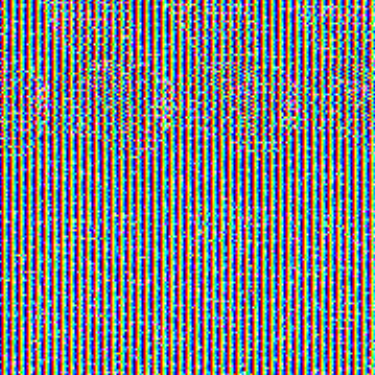} &
        \includegraphics[width=0.125\textwidth]{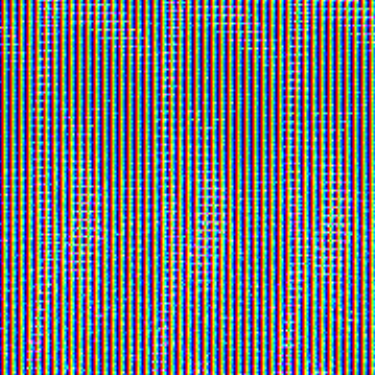}&
        \includegraphics[width=0.125\textwidth]{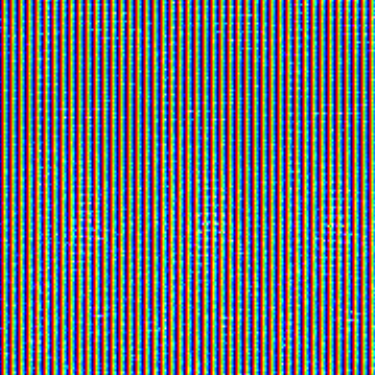} &
        \includegraphics[width=0.125\textwidth]{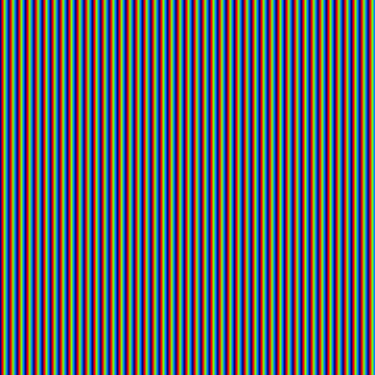} \\
    NUFFT & GRASP & FMLP& InstantNGP & GridTD & Original\\
        \includegraphics[width=0.125\textwidth]{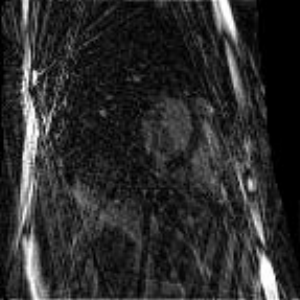} &
        \includegraphics[width=0.125\textwidth]{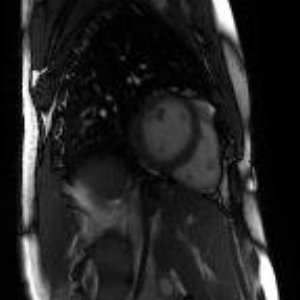} &
        \includegraphics[width=0.125\textwidth]{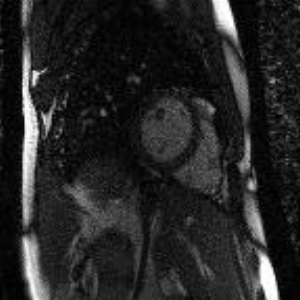} &
        \includegraphics[width=0.125\textwidth]{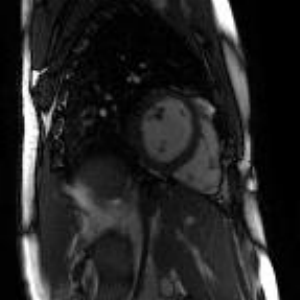}&
        \includegraphics[width=0.125\textwidth]{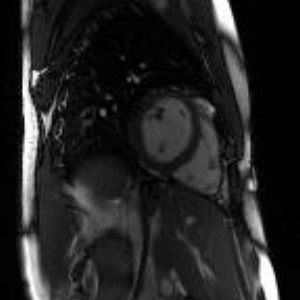} &
        \includegraphics[width=0.125\textwidth]{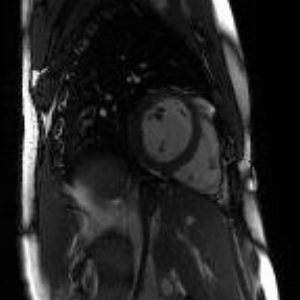} \\
        \includegraphics[width=0.125\textwidth]{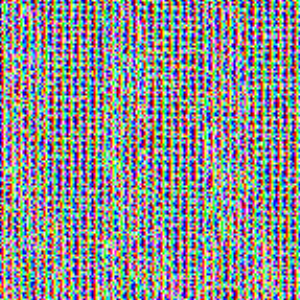} &
        \includegraphics[width=0.125\textwidth]{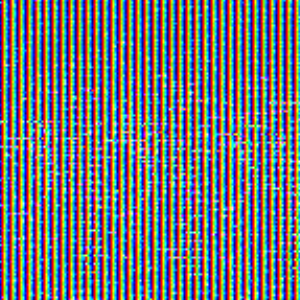} &
        \includegraphics[width=0.125\textwidth]{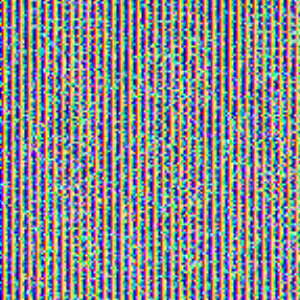} &
        \includegraphics[width=0.125\textwidth]{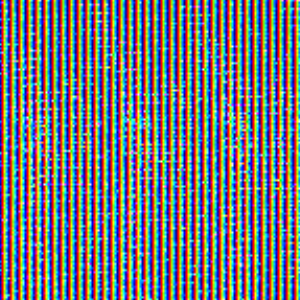}&
        \includegraphics[width=0.125\textwidth]{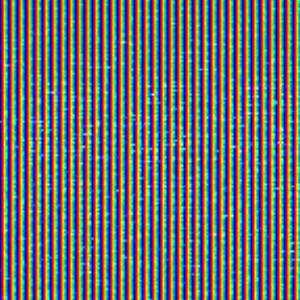} &
        \includegraphics[width=0.125\textwidth]{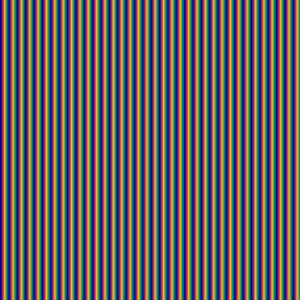} \\\vspace{-0.2cm}
    NUFFT & GRASP & FMLP& InstantNGP & GridTD & Original
    \end{tabular}
    \caption{
Reconstruction comparison using various methods for 21 spokes per frame (acceleration factor = 9.7) and 13 spokes per frame (acceleration factor = 15.7) on the cardiac cine dataset. The first and third rows show the reconstructed MRIs, and the second and fourth rows show the corresponding error maps.
}\vspace{-0.1cm}
    \label{fig:MRI1}
\end{figure*}
\begin{table*}[t]
\tabcolsep=8pt
\centering
\caption{Average quantitative results on the cardiac cine dataset for dynamic MRI reconstruction. (PSNR/SSIM)}
\label{tab:comparison}
\begin{tabular}{cccccccccc}
\hline
\multirow{2}{*}{Method} & \multicolumn{7}{c}{Acceleration factor (AF)} &\multirow{2}{*}{Average time}\\
\cline{2-8}
 & 9.7 & 15.7 & 25.5 & 40.8 & 67.9 & 101.9 & 203.7 \\
\hline
NUFFT & 27.70/0.666 & 24.47/0.536 & 22.21/0.444 & 20.68/0.395 & 19.51/0.353&18.91/0.351&18.12/0.351 & 0.467s\\
GRASP & 38.99/0.955 & 38.28/0.948 & 36.66/0.935 & 34.09/0.897 & 30.14/0.820&24.41/0.671&18.82/0.427 & 5.10s\\
FMLP & 37.84/0.935 & 35.96/0.894 & 33.55/0.831 & 31.30/0.754 & 29.03/0.668&26.77/0.557&\underline{24.88}/0.469 &659.07s \\
InstantNGP & \underline{42.38}/\underline{0.979} & \underline{39.78}/\underline{0.965} & \underline{37.23}/\underline{0.944} & \underline{34.53}/\underline{0.914} & \underline{31.07}/\underline{0.852} &\underline{28.36}/\underline{0.788}&24.80/\textbf{0.676} &654.16s\\
\textbf{GridTD} & \textbf{43.98}/\textbf{0.983} & \textbf{42.22}/\textbf{0.975} & \textbf{39.61}/\textbf{0.962} & \textbf{36.03}/\textbf{0.932} & \textbf{32.27}/\textbf{0.867} &\textbf{29.43}/\textbf{0.791}&\textbf{24.96}/\underline{0.643} &322.56s\\
\hline
\end{tabular}\vspace{-0.3cm}
\end{table*}
We select four representative baselines, including the NUFFT that provides direct zero-filled reconstructions without additional processing; the GRASP \cite{feng2014golden} that incorporates temporal frame grouping, time-averaged coil sensitivity estimation, and TV regularization; the FMLP \cite{kunz2024implicit} that employs INR and Fourier feature encoding for reconstruction; and the InstantNGP \cite{feng2025spatiotemporal} that employs the multi-resolution hash grid encoding for dynamic MRI reconstruction.
\subsubsection{Results}
The quantitative results for compressive dynamic MRI reconstruction are shown in Table~\ref{tab:comparison}. The proposed GridTD method outperforms competing methods across all evaluated settings in PSNR and SSIM metrics. Compared to other unsupervised deep learning methods FMLP and InstantNGP, our GridTD is over 50\% faster due to the compact model structure. In Fig. \ref{fig:MRI1}, we present some visualizations of the reconstructed MRIs and the corresponding error maps, which demonstrate that GridTD better preserves structural edges and robustly recovers fine details, showing the superiority of our method for dynamic MRI reconstruction. 
All these results on video SCI, spectral SCI, and dynamic MRI reconstruction validate the strong performance of GridTD, positioning it as a versatile and state-of-the-art approach for compressive imaging reconstruction.
\subsection{Ablation Study and Parameter Analysis}
We conduct comprehensive ablation studies and parameter analyses to analyze the impact of critical hyperparameters and components in GridTD. All experiments are performed on the \textbf{Runner} scene for the video SCI task. For parameter analysis, we test the hyperparameters including the feature vector dimension $F$ (Table \ref{table:feature}), the number of resolutions $L$ (Table \ref{table:level}), and the TV and SSTV balance parameters $\lambda_1$ and $\lambda_2$ (Table \ref{table:tv} and Table \ref{table:sstv}). We note that varying the feature dimension $F$ and the number of resolutions $L$ is equivalent to varying the CP rank of the GridTD model, hence the parameter sensitivity analysis for the CP rank parameter is omitted. For each test, the optimal performance is achieved at $F=60$, $L=60$, $\lambda_1 = 5\lambda$, and $\lambda_2 = 3.5\lambda$, respectively, where $\lambda=4 \times 10^{-6}$. These parameter settings are expected to strike a favorable balance between expressiveness and stability. We note that varying these parameters within a small range does not abruptly influence the performance, as demonstrated in the parameter analyses. Hence, the GridTD model is quite robust to the variations of these parameters. 

The ablation studies are conducted to verify the contributions of TV, SSTV, tensor decomposition (TD), temporal affine adapter, and the MLP depth. Without TD, the proposed GridTD degenerates to InstantNGP. As shown in Table~\ref{table:structure}, removing either the TV,  SSTV, TD, or affine adapter leads to performance degradation, demonstrating the contributions of all modules introduced in our work. Table~\ref{table:structure} also shows that using a lightweight MLP with depth 2 is sufficient to offer a good trade-off between performance and model size, further demonstrating the strong representation capacity of the GridTD encoding model with only a lightweight MLP.
\begin{table}[t]
\centering
\scriptsize
\setlength{\tabcolsep}{4pt}
\caption{Sensitivity analysis for the feature dimension \( F \) in the GridTD model \eqref{gtdmlp}.}
\label{table:feature}
\begin{tabular}{c|ccccccccc}
\hline
Dataset & \multicolumn{9}{c}{Video SCI on Runner} \\
\hline 
\( F \) &20&30&40&50&60&70&80&90&100  \\
\hline
PSNR     & 34.98 & 34.96 & 34.71 & 34.86 &35.22& 35.10 & 34.97 & 34.93 & 34.99 \\
\hline
SSIM     & 0.954 & 0.955 & 0.954 & 0.954 & 0.957&0.957 & 0.956 & 0.956 & 0.956 \\
\hline
\end{tabular}
\end{table}
\begin{table}[t]
\centering
\scriptsize
\setlength{\tabcolsep}{4pt}
\caption{Sensitivity analysis for the  the resolution number \( L \) in the GridTD model \eqref{gtdmlp}.}
\label{table:level}
\begin{tabular}{c|ccccccccc}
\hline
Dataset & \multicolumn{9}{c}{Video SCI on Runner} \\
\hline 
\( L \) &\hspace{5pt}20&30&40&50&60&70&80&90&100  \\
\hline
PSNR     & \hspace{5pt}34.94 & 35.08 & 35.02 & 35.03 &35.22& 35.20 & 34.93 & 35.12 & 34.92 \\
\hline
SSIM     & \hspace{5pt}0.955 & 0.956 & 0.954 & 0.956 &0.957& 0.957 & 0.955 & 0.957 & 0.955 \\
\hline
\end{tabular}
\vspace{-0.3cm}
\end{table}
\begin{table}[t]
\centering
\scriptsize
\setlength{\tabcolsep}{4pt}
\caption{Sensitivity analysis for the TV balance parameter $\lambda_1 = p_1\lambda$ in model \eqref{eq:v_subproblem_tensor}, where $\lambda$ is fixed at $4 \times 10^{-6}$.}
\label{table:tv}
\begin{tabular}{c|ccccccccc}
\hline
Dataset & \multicolumn{9}{c}{Video SCI on Runner} \\
\hline 
$p_1$ &3&3.5&4&4.5&5&5.5&6&6.5&7  \\
\hline
PSNR     & 35.09 & 34.93 & 35.05 & 35.13 & 35.22&34.91&34.97&34.92&35.18 \\
\hline
SSIM    & 0.955 & 0.955 & 0.956 & 0.955 &0.957& 0.956&0.956&0.955&0.957 \\
\hline
\end{tabular}
\end{table}
\begin{table}[h]
\centering
\scriptsize
\setlength{\tabcolsep}{4pt}
\caption{Sensitivity analysis for the SSTV balance parameter $\lambda_2=p_2\lambda$ in model \eqref{eq:v_subproblem_tensor}, where $\lambda$ is fixed at $4 \times 10^{-6}$.}
\label{table:sstv}
\begin{tabular}{c|cccccccccc}
\hline
Dataset & \multicolumn{9}{c}{Video SCI on Runner} \\
\hline 
$p_2$ &1.5&2&2.5&3&3.5&4&4.5&5&5.5  \\
\hline
PSNR     & 35.12 &35.11& 35.15 & 35.06 & 35.22 & 35.00 & 34.99 & 34.99 & 35.00 \\
\hline
SSIM     &0.955& 0.955 & 0.955 & 0.956 & 0.957 & 0.956 & 0.956& 0.956 & 0.957 \\
\hline
\end{tabular}
\end{table}
\begin{table}[!h]
\centering
\scriptsize
\setlength{\tabcolsep}{5pt}
\renewcommand{\arraystretch}{0.8}
\caption{Ablation study for the TV, SSTV, tensor decomposition (TD), affine adapter, and MLP depth (number of layers) in the GridTD model \eqref{gtdmlp}.}
\label{table:structure}
\begin{tabular}{ccccc|cc}
\hline
 \multicolumn{5}{c|}{Component configuration} & \multicolumn{1}{c}{PSNR} & \multicolumn{1}{c}{SSIM} \\
\hline
 TV & SSTV & TD & Affine & MLP Layers & Value & Value  \\
\hline
\ding{51} & \ding{51}& \ding{51} & \ding{51} & 2 & 35.22  & 0.957\\
\ding{55} & \ding{51}& \ding{51} & \ding{51} & 2 & 34.79 & 0.950  \\
\ding{51} & \ding{55}& \ding{51} & \ding{51} & 2 & 34.76 & 0.950  \\
\ding{51} & \ding{51}& \ding{55} & \ding{51} & 2 & 34.50 & 0.955  \\
\ding{51} & \ding{51}& \ding{51} & \ding{55} & 2 & 35.05  & 0.957 \\
\ding{51} & \ding{51}&\ding{51} & \ding{51} & 3 & 35.15  & 0.956 \\
\ding{51} & \ding{51}&\ding{51} & \ding{51} & 4 & 35.20 & 0.957 \\
\hline
\end{tabular}
\vspace{-0.3cm}
\end{table}
\section{Conclusions}
We have proposed the tensor decomposed multi-resolution grid encoding method for compressive imaging reconstruction, which holds superior representation abilities and efficiency balance over existing unsupervised models due to the integration of tensor decomposition and multi-resolution grids. We theoretically demonstrate the advantages of GridTD from the Lipschitz and generalization bound perspectives. Extensive experiments show that the proposed GridTD achieves state-of-the-art reconstruction performance in compressive dynamic MRI reconstruction, video SCI, and spectral SCI, demonstrating its practical value in real-world compressive imaging problems. In future work, we can apply the GridTD model to more compressive imaging tasks such as hyperspectral-multispectral image fusion and multi-dimensional image super-resolution, to broaden its potential.
\bibliographystyle{IEEEtran}
\bibliography{IEEEabrv}
\end{document}